\documentclass[journal]{IEEEtran}

\usepackage{amsmath,amssymb,epsfig,psfrag,cite,subfigure}
\include{macros}
\usepackage{graphicx}
\usepackage{epstopdf}

\usepackage{acronym}

\usepackage{bm}

\usepackage{pifont}

\usepackage{tikz}
\usetikzlibrary{calc}

\usepackage{tcolorbox}

\usepackage{accents}

\allowdisplaybreaks

\usepackage{soul}

\usepackage{amsthm}

\usepackage{nicefrac}

\usepackage{lipsum,multicol}

\usepackage{algorithm}
\usepackage{algpseudocode}

\algdef{SE}[SUBALG]{Indent}{EndIndent}{}{\algorithmicend\ }%
\algtext*{Indent}
\algtext*{EndIndent}

\usepackage{enumitem}
\usepackage{mwe}    
\usepackage{epsfig,psfrag}
\usepackage{subfigure}
\usepackage{color}
\usepackage{url}
\usepackage{mathtools,xparse}

\usepackage{gensymb}

\usepackage[multiple]{footmisc}

\usepackage{xr}

\makeatletter
\newcommand*{\addFileDependency}[1]{
  \typeout{(#1)}
  \@addtofilelist{#1}
  \IfFileExists{#1}{}{\typeout{No file #1.}}
}
\makeatother

\newcommand*{\myexternaldocument}[1]{%
    \externaldocument{#1}%
    \addFileDependency{#1.tex}%
    \addFileDependency{#1.aux}%
}

\myexternaldocument{supp}

\usepackage{xcolor}

\newcommand\abs[1]{\big|#1\big|}
\newcommand\abss[1]{\lvert#1\rvert}
\newcommand\norm[1]{\left\lVert #1\right\rVert}

\newcommand{\thn}[1]{ {#1^{\rm{th} } } }

\newcommand{\Eee}{\mathbb{E}}

\newcommand{\msens}{\mathcal{S}}
\newcommand{\mcom}{\mathcal{C}}

\newcommand{\mest}{\mathcal{E}}

\newcommand{\snrmean}{{\rm{SNR}}_{\rm{mean}}}
\newcommand{\snr}{{\rm{SNR}}}
\newcommand{\pfa}{P_{\rm{fa}}}
\newcommand{\pd}{P_{\rm{d}}}

\newcommand{\sigmaakb}{\sigma_{\alpha_{b,k}}}
\newcommand{\sigmaaab}{\sigma_{\alpha_b}}

\newcommand{\Tsym}{ T_{\rm{sym}} }

\newcommand{\Tcp}{ T_{\rm{cp}} }

\newcommand{\FF}{ \mathbf{F} }

\newcommand{\deltaf}{ \Delta f }

\newcommand{\fc}{ f_c }
\newcommand{\Ntx}{ N_{\rm{T}} }
\newcommand{\Nrx}{ N_{\rm{R}} }

\newcommand{\atx}{ \aaa_{\rm{T}} }
\newcommand{\arx}{ \aaa_{\rm{R}} }

\newcommand{\quot}[1]{``{#1}''}

\newcommand{\sm}{ s_{m} }
\newcommand{\hhhat}{ \widehat{\hh} }

\newcommand{\Ktilde}{\widetilde{K}}

\newcommand{\alphat}{ \widetilde{\alpha} }
\newcommand{\taut}{ \widetilde{\tau} }
\newcommand{\nut}{ \widetilde{\nu} }
\newcommand{\thetat}{  \widetilde{\theta} }
\newcommand{\dtilde}{ \widetilde{d} }
\newcommand{\sigmarcskt}{{ \widetilde{\sigma}_{{\rm{rcs}},k} }}

\newcommand{\mM}{{ \mathcal{M} }}

\newcommand{\xnm}{ x_{n,m} }

\newcommand{\hnm}{ h_{n,m} }

\newcommand{\ns}{ N_{\mathrm{s}} }

\newcommand{\boldHhat}{ \widehat{\boldH} }

\newcommand{\Rhat}{ \widehat{R} }
\newcommand{\vhat}{ \widehat{v} }

\newcommand{\rhosq}{ \sqrt{\rho} }
\newcommand{\rhosqq}{ \sqrt{1-\rho} }
\newcommand{\ffms}{ \ff_{m,{\rm{s}}} }
\newcommand{\ffmc}{ \ff_{m,{\rm{c}}} }

\newcommand{\fft}{ \widetilde{\ff} }

\newcommand{\sdelta}{\sss_{\Delta}}

\newcommand{\alphahat}{ \widehat{\alpha} }

\newcommand{\alphahatb}{ \boldsymbol{\alphahat}}
\newcommand{\alphab}{ \boldsymbol{\alpha}}

\newcommand{\tauhatbi}[2]{  \tauhat_{#1,#2} }
\newcommand{\nuhatbi}[2]{  \nuhat_{#1,#2} }
\newcommand{\thetahatbi}[2]{  \thetahat_{#1,#2} }
\newcommand{\alphahatbi}[2]{  \alphahat_{#1,#2} }

\newcommand{\ycom}{y^{{\rm{com}}}}
\newcommand{\hhcom}{\hh^{{\rm{com}}}}

\newcommand{\boldYcom}{\mathbf{Y}^{{\rm{com}}}}

\newcommand{\mtCN}{{\mathcal{CN}}}

\newcommand{\thetamax}{\theta_{\rm{max}}}

\newcommand{\vecc}[1]{ {\rm{vec}}\left(#1\right)  }
\newcommand{\veccs}[1]{ {\rm{vec}}\big(#1\big)  }

\newcommand{\diag}[1]{ {\rm{diag}}\left(#1\right)  }

\newcommand{\Imatrix}{{ \boldsymbol{\mathrm{I}} }}

\newcommand{\bphi}{ \bm{\phi} }

\newcommand{\aaa}{\mathbf{a}}
\newcommand{\cc}{ \mathbf{c} }
\newcommand{\bb}{ \mathbf{b} }
\newcommand{\nn}{ \mathbf{n} }

\newcommand{\sigmacsq}{{ \sigma^2_{\rm{c}} }}

\newcommand{\nuhat}{{ \widehat{\nu} }}
\newcommand{\tauhat}{{ \widehat{\tau} }}

\newcommand{\thetahat}{{ \widehat{\theta} }}

\newcommand{\boldzero}{{ {\boldsymbol{0}} }}

\newcommand{\Pt}{{ P_{\rm{T}} }}

\newcommand{\sigmarcsk}{{ \sigma_{{\rm{rcs}},k} }}

\newcommand{\ff}{\mathbf{f}}

\newcommand{\boldY}{ \mathbf{Y} }

\newcommand{\hhatdd}{ \boldHhat^{\rm{DD}} }

\newcommand{\boldX}{ \mathbf{X} }

\newcommand{\boldW}{ \mathbf{W} }

\newcommand{\boldZ}{ \mathbf{Z} }

\newcommand{\boldH}{ \mathbf{H} }

\newcommand{\boldN}{ \mathbf{N} }

\newcommand{\boldD}{ \mathbf{D} }

\newcommand{\boldA}{ \mathbf{A} }

\newcommand{\yy}{ \mathbf{y} }
\newcommand{\yyhat}{ \widehat{\yy}}

\newcommand{\hh}{ \mathbf{h} }

\newcommand{\zz}{ \mathbf{z} }

\newcommand{\kthh}{\mathcal{K}_{(\tauhat, \nuhat)}}

\newcommand{\hhbar}{ \widebar{\hh} }

\newcommand{\Xcal}{\mathcal{X}}

\newcommand{\sss}{ \mathbf{s} }

\newcommand{\sssh}{ \widehat{\mathbf{s}} }

\newcommand{\CC}{ \mathbf{C} }

\newcommand{\transpose}[1]{ {#1}^{T} }

\newcommand{\complexset}[2]{ \mathbb{C}^{#1 \times #2}  }
\newcommand{\complexsett}{ \mathbb{C}  }
\newcommand{\realset}[2]{ \mathbb{R}^{#1 \times #2}  }

\newcommand{\conj}{ {\ast} }

\makeatletter \renewcommand\d[1]{\ensuremath{%
		\;\mathrm{d}#1\@ifnextchar\d{\!}{}}}
\makeatother

\makeatletter
\newcommand*\rel@kern[1]{\kern#1\dimexpr\macc@kerna}
\newcommand*\widebar[1]{%
  \begingroup
  \def\mathaccent##1##2{%
    \rel@kern{0.8}%
    \overline{\rel@kern{-0.8}\macc@nucleus\rel@kern{0.2}}%
    \rel@kern{-0.2}%
  }%
  \macc@depth\@ne
  \let\math@bgroup\@empty \let\math@egroup\macc@set@skewchar
  \mathsurround\z@ \frozen@everymath{\mathgroup\macc@group\relax}%
  \macc@set@skewchar\relax
  \let\mathaccentV\macc@nested@a
  \macc@nested@a\relax111{#1}%
  \endgroup
}
\makeatother

\theoremstyle{remark}

\newtheoremstyle{mytheoremstyle} 
    {\topsep}                    
    {\topsep}                    
    {\upshape}                   
    {.5em}                           
    {\itshape}                   
    {.}                          
    {.5em}                       
    {}  

\theoremstyle{plain}

\newtheoremstyle{iremark}
  {\topsep}   
  {\topsep}   
  {\upshape}  
  {0.2in}       
  {\itshape}  
  {.}         
  {5pt plus 1pt minus 1pt} 
  {\thmname{#1}\thmnumber{ \itshape#2}\thmnote{ (#3)}} 

\newtheorem{theorem}{Theorem}
\newtheorem{lemma}[theorem]{Lemma}

\theoremstyle{definition}
\newtheorem*{proof}{Proof}

\theoremstyle{definition}

\acrodef{RIS}{reconfigurable intelligent surface}
\acrodef{SNR}{signal-to-noise ratio}
\acrodef{ISAC}{integrated sensing and communication}
\acrodef{ISLAC}{integrated sensing, localization, and communication}
\acrodef{LOS}{line-of-sight}
\acrodef{NLOS}{non-line-of-sight}
\acrodef{AOA}{angle-of-arrival}
\acrodef{AOD}{angle-of-departure}
\acrodef{UE}{user equipment}
\acrodef{NF}{near-field}
\acrodef{BS}{base station}
\acrodef{MCRB}{misspecified Cram\'{e}r-Rao bound}
\acrodef{CRB}{Cram\'{e}r-Rao bound}
\acrodef{LB}{lower bound}
\acrodef{ML}{maximum-likelihood}
\acrodef{MML}{mismatched maximum-likelihood}
\acrodef{DL}{downlink}
\acrodef{UL}{uplink}
\acrodef{MIMO}{multiple-input multiple-output}
\acrodef{MISO}{multiple-input single-output}
\acrodef{SISO}{single-input single-output}
\acrodef{SIP}{shift invariance property}
\acrodef{FIM}{Fisher information matrix}
\acrodef{RMSE}{root mean-squared error}
\acrodef{AWGN}{additive white Gaussian noise}
\acrodef{ADMM}{alternating direction method of multipliers}
\acrodef{LS}{least-squares}
\acrodef{SOC}{second-order cone}
\acrodef{CFO}{carrier frequency offset}
\acrodef{i.i.d.}{independently and identically distributed}
\acrodef{MI}{mutual information}
\acrodef{SAC}[S\&C]{sensing and communication}

\acrodef{ULA}{uniform linear array}

\graphicspath{{./Figures/}}

\begin{document}
\bstctlcite{IEEEexample:BSTcontrol}

\title{Fundamental Trade-Offs in Monostatic ISAC: \\A Holistic Investigation Towards 6G}

\author{Musa Furkan Keskin, \textit{Member, IEEE}, Mohammad Mahdi Mojahedian, \textit{Member, IEEE}, Jesus O. Lacruz, \\ Carina Marcus, Olof Eriksson, Andrea Giorgetti, \textit{Senior Member, IEEE}, \\Joerg Widmer, \textit{Fellow, IEEE}, Henk Wymeersch, \textit{Fellow, IEEE}
\thanks{This work is supported, in part, by the Vinnova RADCOM2 project under Grant 2021-02568, by the SNS JU project 6G-DISAC under the EU's Horizon Europe research and innovation program under Grant Agreement No 101139130, and by the European Union under the Italian National Recovery and Resilience Plan (NRRP) 
of NextGenerationEU, partnership on \quot{Telecommunications of the Future} (PE00000001 - program \quot{RESTART}).
}
}

\maketitle


\begin{abstract}

     This paper undertakes a holistic investigation of two fundamental trade-offs in monostatic OFDM integrated sensing and communication (ISAC) systems—namely, the time-frequency trade-off and the spatial trade-off, originating from the choice of modulation order for random data and the design of beamforming strategies, respectively. To counteract the elevated side-lobe levels induced by varying-amplitude data in high-order QAM signaling, we propose a novel linear minimum mean-squared-error (LMMSE) estimator, capable of maintaining robust sensing performance across a wide range of SNRs. Moreover, we explore spatial domain trade-offs through two ISAC transmission strategies: concurrent, employing joint beams, and time-sharing, using separate, time-non-overlapping beams for sensing and communications. Simulations demonstrate superior performance of the LMMSE estimator, especially in detecting weak targets in the presence of strong ones with high-order QAM, consistently yielding more favorable ISAC trade-offs than existing baselines under various modulation schemes, SNR conditions, RCS levels and transmission strategies. We also provide experimental results to validate the effectiveness of the LMMSE estimator in reducing side-lobe levels, based on real-world measurements.

	\vspace{-0.1in}
\end{abstract}

\section{Introduction}
\subsection{Background and Motivation}
As research and standardization efforts for 6G intensify, integrated sensing and communications (ISAC) stands out as a key enabler that can 
facilitate high-quality connectivity and endow networks with intrinsic sensing capability \cite{jointRadCom_review_TCOM,Eldar_SPM_JRC_2020,JCAS_Survey_2022,LiuSeventy23}. This convergence promises to revolutionize network functionalities that extend beyond traditional communication-only paradigms by making sensing a fundamentally built-in component rather than an add-on feature \cite{Fan_ISAC_6G_JSAC_2022,sensingService6G_2024}. In ISAC configurations, the monostatic approach is gaining momentum, especially given its ability to exploit the entire communication data for sensing purposes and its potential to mitigate the stringent synchronization challenges arising in bistatic and multistatic systems \cite{5G_NR_JRC_analysis_JSAC_2022,Barneto2022TCOM,OFDM_PN_TSP_Exploitation_2023,pegoraro2024jump}. Among various waveform candidates for ISAC implementation, the orthogonal frequency-division multiplexing
(OFDM) waveform emerges as a natural choice owing to its widespread adoption across current wireless standards such as 5G, 5G-Advanced, WiFi/WLAN and DVB-T \cite{3gpp_r18_commMag}. The inherent characteristics of OFDM, including its high spectral efficiency, immunity to multipath effects, high-accuracy and low-complexity radar operation and design flexibility, make it an ideal choice for ISAC applications \cite{RadCom_Proc_IEEE_2011,ICI_OFDM_TSP_2020,OFDM_DFRC_TSP_2021,Liu_Reshaping_OFDM_ISAC_2023}. 

In the context of monostatic ISAC, OFDM radar sensing with 5G waveforms has recently received increasing attention \cite{5G_NR_JRC_analysis_JSAC_2022,PRS_ISAC_5G_TVT_2022,li2023isac,5g_sensing_TAES_2023,refSig_6G_ISAC_2023,MIMO_OFDM_3D_estimation_2023,ofdm_isac_silvia_2024}. Evaluations have been conducted on the sensing performance of a base station (BS) acting as a monostatic radar with downlink 5G new radio (NR) OFDM signals, focusing on the use of pilot/reference symbols alone \cite{PRS_ISAC_5G_TVT_2022,refSig_6G_ISAC_2023} as well as a combination of pilot/reference and data symbols \cite{5G_NR_JRC_analysis_JSAC_2022,li2023isac}. To achieve satisfactory target detection performance with sparse pilots, maintaining low sidelobe levels and eliminating ambiguities becomes critical in 5G/6G OFDM sensing \cite{ofdm_isac_silvia_2024}. In \cite{5g_sensing_TAES_2023,PRS_ISAC_5G_TVT_2022}, various 5G signals and channels, including SSB, PRS, PDSCH and CSI-RS, have been investigated in terms of their delay-Doppler ambiguity function (AF) characteristics. Similarly, the study in \cite{li2023isac} evaluates tracking performance using 5G NR signals within V2X networks, leading to guidelines on the design of 5G-Advanced and 6G frame structures \cite{li2023isac,PRS_ISAC_5G_TVT_2022}. Overall, these studies highlight the significant potential of employing OFDM waveforms for sensing purposes, while also noting the challenges associated with target detection due to suboptimal side-lobe performance.

Another major challenge pertaining to side-lobe levels in monostatic OFDM sensing emerges from the use of random communication data \cite{Liu_Reshaping_OFDM_ISAC_2023,mimo_ofdm_isac_sidelobes_2024,subspace_OFDM_ISAC_ICASSP_2024}. This issue unveils a crucial inherent trade-off in monostatic OFDM ISAC systems, namely \textit{the time-frequency trade-off}, falling under the umbrella of deterministic-random trade-offs \cite{LiuSeventy23,drt_isac_2023}. The time-frequency trade-off stems from the choice of modulation order for random data: higher-order QAM enhances communication rates but compromises sensing performance, as the random, non-constant-modulus data leads to elevated side-lobe levels \cite{Liu_Reshaping_OFDM_ISAC_2023}. On the other hand, employing constant-modulus QPSK minimizes side-lobe levels, enhancing sensing performance but at the cost of lower communication rates. To explore and enhance the time-frequency trade-off, various transmit optimization schemes have been recently proposed, focusing on OFDM \cite{mimo_ofdm_isac_sidelobes_2024,Liu_Reshaping_OFDM_ISAC_2023,inputOptOFDM_DFRC_2023} as well as generic communication systems (along the lines of broader deterministic-random trade-offs) \cite{drt_isac_2023,sensingRandom_2023,IT_JCAS_Gaussian_2023}. In \cite{drt_isac_2023,inputOptOFDM_DFRC_2023,IT_JCAS_Gaussian_2023}, multiple ISAC trade-off optimization problems have been formulated to optimize the input data distribution for improving the time-frequency trade-off. Considering MIMO-OFDM ISAC systems, \cite{mimo_ofdm_isac_sidelobes_2024} designs a symbol-level transmit precoder to minimize the range-Doppler integrated side-lobe level under multi-user communication quality-of-service (QoS) constraints. Similarly, the work in \cite{Liu_Reshaping_OFDM_ISAC_2023} proposes a probabilistic constellation shaping (PCS) approach to maximize the achievable rate under constraints on the variance of the radar AF with random OFDM data.

\vspace{-0.1in}
\subsection{Research Gaps}

Despite extensive research into sensing algorithms \cite{5G_NR_JRC_analysis_JSAC_2022,PRS_ISAC_5G_TVT_2022,li2023isac,5g_sensing_TAES_2023,refSig_6G_ISAC_2023,MIMO_OFDM_3D_estimation_2023} and trade-off analysis \cite{mimo_ofdm_isac_sidelobes_2024,Liu_Reshaping_OFDM_ISAC_2023,inputOptOFDM_DFRC_2023,drt_isac_2023,sensingRandom_2023,IT_JCAS_Gaussian_2023} in OFDM ISAC systems, a number of crucial topics remain unexplored. First, previous studies on OFDM sensing have typically utilized QPSK signaling \cite{5G_NR_JRC_analysis_JSAC_2022,mimo_ofdm_isac_sidelobes_2024,PRS_ISAC_5G_TVT_2022,5g_sensing_TAES_2023}, which compromises communication performance, or have implemented QAM data with reciprocal filtering \cite{MIMO_OFDM_3D_estimation_2023,RadCom_Proc_IEEE_2011,OFDM_Radar_Corr_TAES_2020,li2023isac}, which leads to elevated side-lobe levels at low SNRs due to enhanced noise power \cite{mimo_ofdm_isac_sidelobes_2024}. On the other hand, employing matched filtering \cite{OFDM_Radar_Corr_TAES_2020} with QAM signaling as an alternative to reciprocal filtering results in increased side-lobe levels at high SNRs, caused by varying-amplitude data. Hence, no universally effective sensing algorithm exists that maintains robust performance across a wide range of SNR conditions, particularly when dealing with high-order modulations in monostatic OFDM ISAC systems. Moreover, prior works on time-frequency trade-off analysis focus exclusively on ISAC \textit{transmit} signal optimization \cite{mimo_ofdm_isac_sidelobes_2024,Liu_Reshaping_OFDM_ISAC_2023,inputOptOFDM_DFRC_2023,drt_isac_2023,sensingRandom_2023,IT_JCAS_Gaussian_2023} without tackling the problem of sensing \textit{receiver} design. This leads to several drawbacks as optimizing transmitter design to tune ISAC trade-offs may hamper communication functionality and necessitate conveying additional control information to the communication receiver for symbol decoding, thereby increasing overhead. Conversely, the receiver design strategy adopted in this paper aligns with the \textit{opportunistic sensing} paradigm \cite{grossi2020adaptive}, which does not interfere with the communication system. Finally, the \textit{spatial trade-off}, resulting from the design of ISAC beamformers (e.g., \cite{5G_NR_JRC_analysis_JSAC_2022}), has been rarely studied alongside the time-frequency trade-off, which merits further investigation considering their complex interplay within various transmission schemes, such as \textit{concurrent} and \textit{time-sharing} \cite{IT_JCAS_2022}.

In light of the existing literature on monostatic ISAC systems, several fundamental questions remain unanswered:
\begin{itemize}
    \item How can robust sensing algorithms be designed for monostatic OFDM ISAC systems, accommodating known yet \textit{varying-amplitude random data symbols} in the time-frequency domain generated by high-order modulation schemes? To what extent can we mitigate the \textit{masking effect} in both single- and multi-target environments, resulting from increased side-lobe levels induced by non-constant-modulus data?

    \item What spatial domain trade-offs arise from the choice of ISAC transmit beamformers? How do different transmission strategies, namely, \textit{concurrent and time-sharing transmission}, perform under various settings concerning the sensing algorithm, percentage of sensing-dedicated pilots and modulation order of random data? 
    \item Given the inherent tension between \ac{SAC} in time-frequency and spatial domains, how can communication data rates and accompanying ISAC trade-offs be realistically evaluated under different modulations and power allocations between \ac{SAC} beams?
\end{itemize}

\begin{figure*}
	\centering
    \vspace{-0.1in}
	\includegraphics[width=0.9\linewidth]{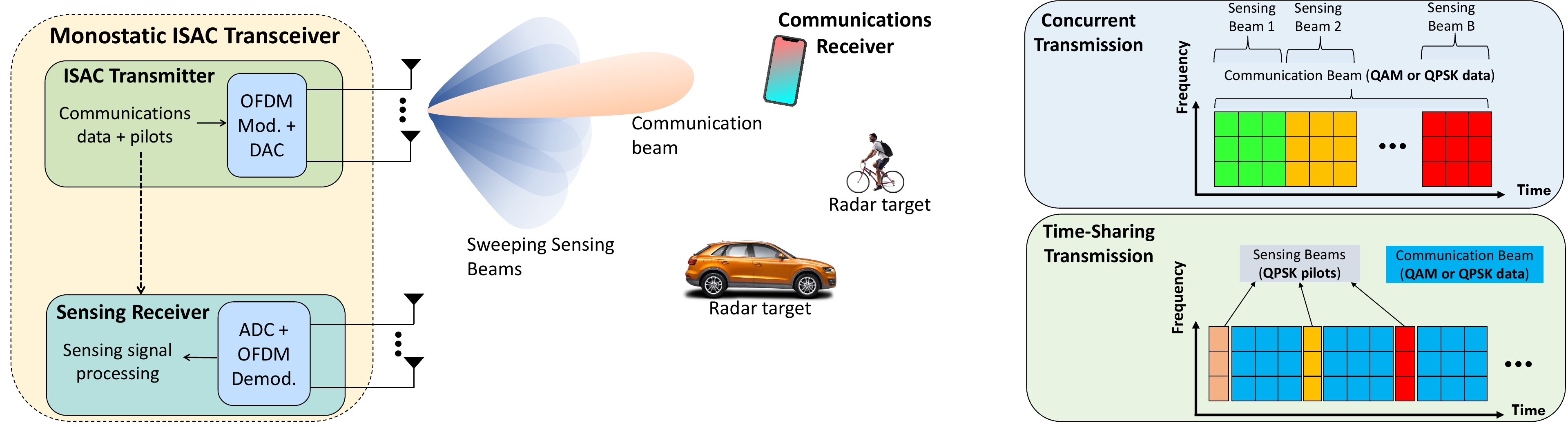}
	\vspace{-0.1in}
	\caption{Monostatic ISAC system featuring a monostatic ISAC transceiver, which integrates an ISAC transmitter and a sensing receiver on the same hardware platform, and a communications receiver on a separate device. In concurrent transmission, sensing receiver utilizes all communication data, with transmit power distributed between sensing and communication beams. In time-sharing transmission, dedicated sensing and communications beams are transmitted in a time-multiplexed fashion.} 
	\label{fig_isac_system}
	\vspace{-5mm}
\end{figure*}

\vspace{-0.1in}
\subsection{Contributions}
With the aim of addressing the identified research gaps towards 6G, this study performs a comprehensive investigation of fundamental trade-offs in monostatic OFDM ISAC systems by introducing novel sensing algorithms, beamforming/transmission strategies and holistic performance evaluations using both simulated and experimental data. The main contributions can be summarized as follows:
\begin{itemize}
    \item \textbf{Sensing with Varying-Amplitude Data in OFDM ISAC (Time-Frequency Trade-Offs):} We introduce a novel MIMO-OFDM radar sensing algorithm based on an LMMSE estimate of the time-frequency domain radar channel, specifically tailored for detection with varying-amplitude random communication data in monostatic ISAC systems. The proposed algorithm can effectively suppress high side-lobes induced by varying-amplitude data and significantly improve detection capability for weak targets, outperforming the conventional OFDM sensing algorithms (i.e., matched and reciprocal filtering \cite{MIMO_OFDM_3D_estimation_2023,RadCom_Proc_IEEE_2011,OFDM_Radar_Corr_TAES_2020,5G_NR_JRC_analysis_JSAC_2022}).

\item \textbf{ISAC Transmission Strategies  (Spatial Trade-Offs):} To analyze the spatial domain trade-offs, we investigate two ISAC transmission strategies: concurrent transmission, where the same beam is used for \ac{SAC}, and time-sharing transmission, which uses dedicated beams for \ac{SAC} non-overlapping in time. We also introduce a mutual information (MI) approximation method that leverages Monte Carlo sampling techniques to evaluate data rates and the resulting ISAC trade-offs under various modulation schemes.
    
    \item \textbf{Holistic Performance Investigation via Simulations and Experiments:} To offer comprehensive guidelines for the design of 6G systems, we carry out extensive simulations to investigate fundamental ISAC trade-offs within the time-frequency and spatial domains under a wide array of transmission settings and channel conditions, including modulation order, SNR and power/time allocation between \ac{SAC}. The proposed LMMSE estimator demonstrates substantial improvements in the trade-offs between probability of detection and achievable rate compared to existing baselines. Moreover, we identify an intricate, scenario-dependent interplay between concurrent and time-sharing strategies, determined by two opposing factors: full-data utilization for sensing and impact of QAM data on sensing, which are beneficial and detrimental, respectively, for the concurrent strategy. We also provide experimental results to verify the effectiveness of the LMMSE estimator on real measurements.

\end{itemize}

\section{System Model and Problem Description}
Consider a monostatic OFDM ISAC system consisting of two entities, as shown in Fig.~\ref{fig_isac_system}: a multiple-antenna dual-functional ISAC transceiver and a single-antenna communications receiver (RX). The ISAC transceiver contains \textit{(i)} an ISAC transmitter (TX) with an $\Ntx$-element \ac{ULA} that sends data/pilot symbols to the communications RX and \textit{(ii)} a sensing RX with an $\Nrx$-element \ac{ULA} that performs radar sensing by processing the backscattered signals for target detection, estimation and tracking \cite{Fan_ISAC_6G_JSAC_2022,JCAS_Survey_2022,RadCom_Proc_IEEE_2011}. In the monostatic configuration under consideration, the sensing RX is co-located on the same device as the ISAC TX, thereby sharing the same oscillator and having access to the entire OFDM transmit signal \cite{dfrc_mimo_ofdm_TSP_2023,5G_NR_JRC_analysis_JSAC_2022,OFDM_PN_TSP_Exploitation_2023}. To ensure that the sensing RX experiences no self-interference during full-duplex operation, we assume that the TX/RX antennas at the ISAC transceiver are sufficiently isolated \cite{RadCom_Proc_IEEE_2011,SI_5G_2015,OFDM_Radar_Phd_2014,80211_Radar_TVT_2018,OFDM_FD_LTE_2019,Fan_ISAC_6G_JSAC_2022,OFDM_5G_6G_Sensing_TWC_2021}. In this section, we introduce the OFDM transmit signal model, different ISAC transmission/beamforming strategies and provide received signal models at both the sensing and communication RXs. Additionally, we formulate the problems of interest tackled throughout the paper.

\vspace{-3mm}
\subsection{Transmit Signal Model}
We consider an OFDM frame with $N$ subcarriers and $M$ symbols. The complex baseband transmit signal for the $\thn{m}$ symbol can be expressed as \cite{RadCom_Proc_IEEE_2011,5G_NR_JRC_analysis_JSAC_2022}
\begin{align} \label{eq_st}
    s_m(t) = \frac{1}{\sqrt{N}} \sum_{n = 0}^{N-1}  \xnm \, e^{j 2 \pi n \deltaf t} g\left(\frac{t - m\Tsym}{\Tsym}\right) \,,
\end{align}
where $\xnm$ is the data/pilot on the $\thn{n}$ subcarrier and the $\thn{m}$ symbol, $\deltaf = 1/T$ is the subcarrier spacing with $T$ representing the elementary symbol duration, $\Tsym = T + \Tcp$ is the total symbol duration including the cyclic prefix (CP) $\Tcp$, and $g(t)$ is a rectangular pulse that takes the value $1$ for $t \in [0, 1]$ and $0$ otherwise. Employing single-stream beamforming \cite{80211_Radar_TVT_2018,MIMO_OFDM_ICI_JSTSP_2021,5G_NR_JRC_analysis_JSAC_2022}, the passband transmit signal over the TX array for the entire OFDM frame is given by
\begin{align}\label{eq_passband_st}
\Re \Big\{ \sum_{m = 0}^{M-1} \ff_m \sm(t) e^{j 2 \pi \fc t} \Big\} \,,
\end{align}
where $\ff_m \in \complexset{\Ntx}{1}$ is the TX beamforming (BF) vector applied for the $\thn{m}$ symbol and $\fc$ is the carrier frequency. Denoting by $\Pt$ the transmit power, we set $\norm{\ff_m}^2 = \Pt \, \forall m$ and $\Eee\{ \abss{\xnm}^2 \} = 1$.
\vspace{-3mm}
\subsection{ISAC Transmission Strategies}\label{sec_isac_str}
We investigate two ISAC transmission strategies concerning the choice of $\xnm$ and $\ff_m$ in \eqref{eq_passband_st}.

\subsubsection{Concurrent Transmission}\label{sec_conc}
In the concurrent transmission, a common beam is utilized simultaneously for \ac{SAC}. Using the multibeam approach, the TX BF vector employed at the $\thn{m}$ symbol is given by \cite{multibeam_JRC_TVT,5G_NR_JRC_analysis_JSAC_2022}
\begin{align} \label{eq_ffm}
    \ff_m = \rhosq \, \ffms + \rhosqq \, \ffmc  \,,
\end{align}
where $\ffms \in \complexset{\Ntx}{1}$ and $\ffmc \in \complexset{\Ntx}{1}$ represent, respectively, the sensing and communication BF vectors, and $0 \leq \rho \leq 1$ denotes the ISAC weight that controls the trade-off between \ac{SAC}. In this strategy, all $\xnm$'s are assumed to be data symbols\footnote{For ease of exposition and analysis, the communication channel is assumed to be perfectly estimated a-priori, eliminating the need for pilot symbols. While the proposed framework can theoretically be extended to account for the impact of communication pilots on ISAC performance trade-offs, exploring this aspect is beyond the scope of the current study and is reserved for future research.} intended for the communications RX, while the sensing RX exploits the entire frame for radar sensing. 

\subsubsection{Time-Sharing Transmission}\label{sec_tst}
In the time-sharing transmission, dedicated beams for sensing and communications are used in non-overlapping time slots \cite{IT_JCAS_2022}. Thus, at each OFDM symbol, the ISAC TX transmits either a sensing beam or a communication beam. More formally, 
\begin{align}\label{eq_tst}
    \ff_m = \begin{cases}
	\ffms,&~ m \in \msens   \\
	\ffmc,&~ m \in \mcom
	\end{cases} \,, ~
 \xnm = \begin{cases}
	{\rm{pilot}},&~ m \in \msens   \\
	{\rm{data}},&~ m \in \mcom
	\end{cases} \, ,
\end{align}
where $\msens \cup \mcom = \{0, \ldots, M-1\}$ and $\msens \cap \mcom = \emptyset$. In \eqref{eq_tst}, `pilot' refers to dedicated \textit{sensing pilots with unit amplitude} \cite{dedicatedSensingISAC_2024}, whereas `data' could be either unit-amplitude (e.g., QPSK) or varying-amplitude (e.g., QAM). The sensing RX utilizes only the pilots for sensing, while the communications rate might be compromised due to the replacement of a portion of data by pilots. The time-sharing ratio $\abss{\msens}/M$ allows tuning the trade-off between \ac{SAC}.

\vspace{-3mm}
\subsection{Sensing Signal Model}
Given 
\eqref{eq_st} and \eqref{eq_passband_st},  the backscattered signal at the sensing RX array over subcarrier $n$ and symbol $m$ after the CP removal and FFT operations can be written as \cite{overview_SP_JCS_JSTSP_2021}
\begin{align} \label{eq_ynm}
    \yy_{n,m} = \boldH_{n,m} \ff_m \xnm + \nn_{n,m} \in \complexset{\Nrx}{1}  \,,
\end{align}
where $\boldH_{n,m} \in \complexset{\Nrx}{\Ntx}$ denotes the channel matrix at the $\thn{n}$ subcarrier and the $\thn{m}$ symbol, and  $\nn_{n,m} \in \complexset{\Nrx}{1}$ is the \ac{AWGN} component with $\nn_{n,m} \sim \mtCN(\boldzero, \sigma^2 \Imatrix)$ and $\sigma^2 = N_0 N \deltaf$, with $N_0$ representing the noise power spectral density (PSD). 
Considering the presence of $K$ point targets in the environment, the sensing channel can be expressed as
\begin{align} \label{eq_sens_channel}
    \boldH_{n,m} = \sum_{k=0}^{K-1} \alpha_k e^{-j 2 \pi n \deltaf \tau_k} e^{j 2 \pi m \Tsym \nu_k} \arx(\theta_k) \atx^T(\theta_k)   \,,
\end{align}
where $\alpha_k$, $\tau_k$, $\nu_k$ and $\theta_k$ denote, respectively, the complex channel gain (including the effects of path attenuation and radar cross section (RCS)), round-trip delay, Doppler shift and \ac{AOA}/\ac{AOD} of the $\thn{k}$ target. Here, $\tau_k = 2 d_k/c$ and $\nu_k = 2 v_k/\lambda$, with $c$, $\lambda = c/\fc$, $d_k$ and $v_k$ denoting the speed of propagation, the wavelength, the range and velocity of the $\thn{k}$ target, respectively. In addition, the channel gain is given by the radar range equation $\abss{\alpha_k}^2 =   \sigmarcsk \lambda^2 / [(4 \pi)^3 d_k^4 ]$ \cite[Eq.~(2.8)]{richards2005fundamentals},
where $\sigmarcsk$ denotes the RCS of the $\thn{k}$ target. Moreover, the ULA steering vectors are defined as
\begin{align}
        \atx(\theta) &= \transpose{ \big[ 1  ~ e^{j \frac{2 \pi}{\lambda} d \sin(\theta)}~ \ldots ~ e^{j \frac{2 \pi}{\lambda} d (\Ntx-1) \sin(\theta)} \big] } \,, \\ \label{eq_ar_steer}
        \arx(\theta) &= \transpose{ \big[ 1  ~ e^{j \frac{2 \pi}{\lambda} d \sin(\theta)}~ \ldots ~  e^{j \frac{2 \pi}{\lambda} d (\Nrx-1) \sin(\theta)} \big] } \,,
\end{align}
where $d = \lambda/2$ denotes the antenna element spacing.

\vspace{-3mm}
\subsection{Communications Signal Model}
The signal received at the communications RX over subcarrier $n$ and symbol $m$ is given by
\begin{align} \label{eq_ynmcom}
    \ycom_{n,m} = (\hhcom_{n,m})^T \ff_m \xnm + z_{n,m} \in \complexsett  ~,
\end{align}
where $\hhcom_{n,m} \in \complexset{\Ntx}{1}$ is the communication channel over subcarrier $n$ and symbol $m$. In addition, $z_{n,m}$ is AWGN with $z_{n,m} \sim \mtCN(0, \sigmacsq)$. 
Assuming  $\Ktilde$ paths between the ISAC TX and the communications RX, $\hhcom_{n,m}$ can be modeled as
\begin{align}\label{eq_hhcom}
    \hhcom_{n,m} = \sum_{k=0}^{\Ktilde-1} \alphat_k e^{-j 2 \pi n \deltaf \taut_k} e^{j 2 \pi m \Tsym \nut_k} \atx(\thetat_k) \,,
\end{align}
where $\alphat_k$, $\taut_k$, $\nut_k$ and $\thetat_k$ denote, respectively, the complex channel gain, delay, Doppler shift and \ac{AOD} of the $\thn{k}$ path. Here, $k=0$ represents the \ac{LOS} path. Accordinly, the channel gains are given by $\abss{\alphat_0}^2 =  \lambda^2 / (4 \pi \dtilde_0)^2 $ and $\abss{\alphat_k}^2 = \sigmarcskt \lambda^2 / [(4 \pi)^3 \dtilde^2_{k,1} \dtilde^2_{k,2}] $ for $k>0$  \cite[Eq.~(45)]{Zohair_5G_errorBounds_TWC_2018}, where $\sigmarcskt$ denotes the RCS of the scatterer associated with the $\thn{k}$ path, $\dtilde_{k,1}$ and $\dtilde_{k,2}$ are the distances between TX-scatterer and scatterer-RX.

\vspace{-3mm}
\subsection{Beamformers for Sensing and Communications}\label{sec_bf_senscom}
To search for potential targets in the environment, the TX sensing beam $\ffms$ in \eqref{eq_ffm} and \eqref{eq_tst} sweeps an angular range $[-\thetamax, \thetamax]$. We assume the use of $B$ different sensing beams over $M$ symbols. Let $\mM_{b}$ denote the set of indices of symbols for which the $\thn{b}$ beam is employed, i.e., $\mM_1 = \{1, \ldots, M/B\}, \, \ldots \,, \mM_{B} = \{M - M/B+1, \ldots, M\}$. Hence, the sensing beams are given by $\ffms = \sqrt{\Pt/\Ntx} \atx^\conj(\theta_b)$ for $m \in \mM_{b}$, where $\theta_b = -\thetamax + 2(b-1)/(B-1) \thetamax $.  

For communications, we assume a \ac{LOS} beam tracking scenario \cite{sensingAssistedComm_TWC_2020} where the TX communication beam $\ffmc$ in \eqref{eq_ffm} and \eqref{eq_tst} is aligned with $\thetat_0$ and stays constant during the entire frame, i.e., $\ffmc = \sqrt{\Pt/\Ntx} \atx^\conj(\thetat_0) \, \forall m$.

\vspace{-3mm}
\subsection{Impact of Non-Constant-Modulus Data on Sensing}
To demonstrate the impact of different modulation orders on sensing performance, we provide an illustrative example of range profiles in Fig.~\ref{fig_range_profile_RF_QPSK_QAM}, obtained with QPSK and $1024-$QAM modulations. The range profiles are generated by applying reciprocal filtering and 2-D FFT based standard OFDM radar processing (e.g., \cite{RadCom_Proc_IEEE_2011,OFDM_Radar_Phd_2014,OFDM_Radar_Corr_TAES_2020,5G_NR_JRC_analysis_JSAC_2022}). From the figure, it is evident that while high-order modulations provide high data rates, they also tend to result in elevated side-lobe levels, which significantly impairs the detection of weak targets. This example sheds light on one of the fundamental trade-offs in ISAC systems, namely the time-frequency trade-off, from the perspective of detection.

\begin{figure}
	\centering
    \vspace{-0.1in}
	\includegraphics[width=0.95\linewidth]{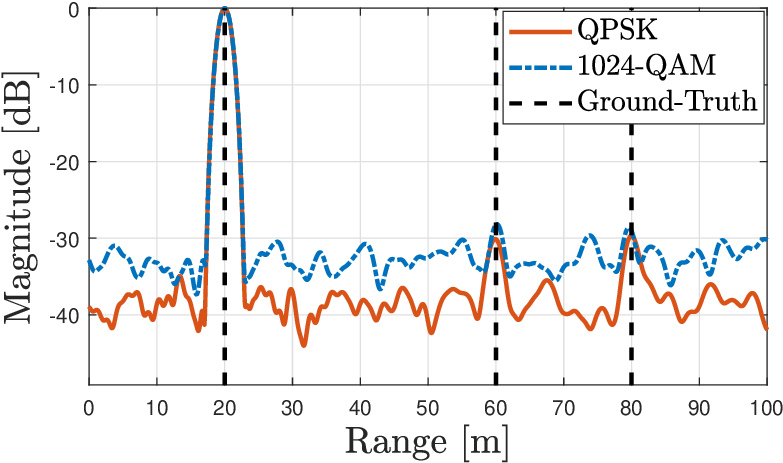}
	\vspace{-0.1in}
	\caption{Range profiles obtained via conventional OFDM radar processing based on reciprocal filtering (i.e., element-wise division of transmit symbols) and 2-D FFT \cite{RadCom_Proc_IEEE_2011,OFDM_Radar_Phd_2014,OFDM_Radar_Corr_TAES_2020,5G_NR_JRC_analysis_JSAC_2022} with QPSK and $1024-$QAM modulations. The parameters employed are $B = 100 \, \rm{MHz}$, $M = 140$, with additional details as specified in Table~\ref{tab_parameters}. The scenario involves three targets, each moving at $10 \, \rm{m/s}$, with RCSs of $(0, -10, -5) \, \rm{dBsm}$ and located at ranges of $(20, 60, 80) \, \rm{m}$ and angles of $(10^\circ, 5^\circ, 15^\circ)$. It is observed that the use of communication data symbols from high-order modulations increases side-lobe levels, thereby masking the presence of weaker targets.} 
	\label{fig_range_profile_RF_QPSK_QAM}
\end{figure}

\vspace{-3mm}
\subsection{Problem Description}
Given the OFDM frame $\boldX \in \complexset{N}{M}$ with $[\boldX]_{n,m} = \xnm$ consisting of random data (and, potentially, pilot symbols), the ISAC TX beamformers $\{\ff_m\}_{m=0}^{M-1}$, and the sensing and communications models in \eqref{eq_ynm} and \eqref{eq_ynmcom}, the problems of interest in this work are: \textit{(i)} to develop a radar sensing algorithm to detect the presence of multiple targets and estimate their delay-Doppler-angle parameters from 3-D tensor observations $\{\yy_{n,m}\}$ in \eqref{eq_ynm} over $\Nrx$ receive antennas, $N$ subcarriers and $M$ symbols, \textit{(ii)} to evaluate the data rate for the signal model in \eqref{eq_ynmcom} across diverse modulation orders (e.g., QAM, QPSK), and \textit{(iii)} to investigate ISAC trade-offs for different configurations of $\boldX$ and different choices of $\ff_m$ (determined by the ISAC weight in \eqref{eq_ffm}, the data/pilot ratio in \eqref{eq_tst} and the modulation order in both strategies).
We begin by tackling the first problem in Sec.~\ref{sec_sens}. The method to approximate the communication MI for data rate evaluation will be introduced in Sec.~\ref{sec_comm_rate}. Finally, we present simulation and experimental results to explore sensing and ISAC trade-off performances in Sec.~\ref{sec_num}.

\section{Radar Sensing Algorithm}\label{sec_sens}
In this section, we introduce a novel sensing algorithm for multiple target detection and accompanying parameter estimation using the observation in \eqref{eq_ynm}, that can account for arbitrary $\boldX$ and TX beam sweeping.

\vspace{-3mm}
\subsection{Beam-Specific Channel Estimation}
Target detection from the observation \eqref{eq_ynm} via \textit{coherent} processing in the spatial-frequency-time domains involves a computationally demanding 3-D search over the delay-Doppler-angle tuples of potential targets \cite{MIMO_OFDM_radar_TAES_2020,MIMO_OTFS_ISAC_2023}. Furthermore, the beam sweeping procedure used for sensing, as discussed in Sec.~\ref{sec_bf_senscom}, results in varying gains for targets across different beams. This is because the overall gain of the $\thn{k}$ target in \eqref{eq_sens_channel} is represented by $\alpha_k \atx^T(\theta_k) \ff_m$, which varies over symbols $m$ when $\ff_m$ changes \cite{5G_NR_JRC_analysis_JSAC_2022}.

Given these two challenges in detection/estimation from \eqref{eq_ynm}, we propose to formulate the problem of sensing as a series of \textit{beam-specific channel estimation} problems where for each beam the time-frequency domain radar channel is estimated per RX element by treating it as an \textit{unstructured channel} rather than a delay-Doppler parameterized one as in \eqref{eq_sens_channel}, followed by \textit{noncoherent} integration \cite{MIMO_OFDM_radar_TAES_2020,MIMO_OTFS_ISAC_2023} of the resulting delay-Doppler images over the RX array. More formally, for the $\thn{b}$ beam, the sensing observations in \eqref{eq_ynm} at the $\thn{i}$ RX element can be expressed using \eqref{eq_sens_channel} as
\begin{align} \label{eq_y_comp_beam}
    \boldY_{i,b} = \boldH_{i,b} \odot \boldX_b  + \boldN_{i,b} \in \complexset{N}{M_b} ~, 
\end{align}
where $\odot$ denotes the Hadamard (element-wise) product,
\begin{align} \label{eq_hib}
    \boldH_{i,b} &\triangleq \sum_{k=0}^{K-1} \alpha_{b,k} \bb(\tau_k) \cc_b^H(\nu_k) [\arx(\theta_k)]_i \in \complexset{N}{M_b} 
\end{align}
represents the time-frequency radar channel at the $\thn{i}$ RX element for the $\thn{b}$ beam, $\boldX_b = [\boldX]_{:,\mM_b} \in \complexset{N}{M_b}$ is the transmit symbols for the $\thn{b}$ beam, $\vecc{\boldN_{i,b}} \sim \mtCN(\boldzero, \sigma^2  \Imatrix)$ and $M_b = \lvert \mM_b \rvert$. In \eqref{eq_hib}, $\alpha_{b,k} = \alpha_k \atx^T(\theta_k) \fft_b$ denotes the overall gain of the $\thn{k}$ target for the $\thn{b}$ beam,
\begin{align} \label{eq_steer_delay}
	\bb(\tau) & \triangleq  \transpose{ \big[ 1 ~ e^{-j 2 \pi \deltaf \tau} ~ \ldots ~  e^{-j 2 \pi (N-1) \deltaf  \tau} \big] } ~, \\ \label{eq_steer_doppler}
	\cc(\nu) & \triangleq \transpose{ \big[ 1 ~ e^{-j 2 \pi  \Tsym \nu } ~ \ldots ~  e^{-j 2 \pi  (M-1) \Tsym \nu } \big] } ~, 
\end{align}
represent the frequency-domain and temporal (slow-time) steering vectors, respectively, and $\cc_b(\nu) = [\cc(\nu)]_{\mM_b}$. Moreover, $\fft_b \in \complexset{\Ntx}{1}$ represents the $\thn{b}$ beam, i.e., $[\FF]_{:,\mM_b} = [\fft_b \, \ldots \, \fft_b ]  \in \complexset{\Ntx}{M_b} $, with $\FF \triangleq [\ff_1 \, \ldots \, \ff_m] \in \complexset{\Ntx}{M}$ denoting the entire TX beamforming matrix.

We revisit the two commonly employed methods to estimate $\boldH_{i,b}$ from $\boldY_{i,b}$ in \eqref{eq_y_comp_beam} and present the proposed method.
\subsubsection{Reciprocal Filtering (RF)}
The reciprocal filtering (RF) performs channel estimation by element-wise division of received symbols by transmit symbols \cite{RadCom_Proc_IEEE_2011,OFDM_Radar_Phd_2014,OFDM_Radar_Corr_TAES_2020,5G_NR_JRC_analysis_JSAC_2022,PRS_ISAC_5G_TVT_2022}
 \begin{align}\label{eq_rf}
        \boldHhat_{i,b}= \boldY_{i,b} \oslash \boldX_b  \,,
\end{align}
where $\oslash$ denotes element-wise division. The RF estimator in \eqref{eq_rf} can be derived as a result of the least-squares (LS) solution in \eqref{eq_y_comp_beam} and thus corresponds to the \textit{zero-forcing} 
\cite{ce_ofdm_95}.  
\subsubsection{Matched Filtering (MF)}
The matched filtering (MF) approach aims to maximize the SNR at the output of the filter and applies conjugate multiplication of received symbols by transmit symbols \cite{OFDM_Radar_Corr_TAES_2020,reciprocalFilter_OFDM_2023}
\begin{align} \label{eq_mf}
        \boldHhat_{i,b} = \boldY_{i,b} \odot \boldX_b^\conj   \,.
\end{align}

\subsubsection{Proposed LMMSE Estimator}
We propose to employ an LMMSE estimator to estimate $\boldH_{i,b}$ in \eqref{eq_y_comp_beam} by treating it as a random unknown parameter as opposed to deterministic modeling in RF and MF strategies. To this end, let us assume that the vectorized version of the channel in \eqref{eq_hib}, given by
\begin{align} \label{eq_hhib_vec}
    \hh_{i,b} \triangleq \vecc{\boldH_{i,b}} = \sum_{k=0}^{K-1} \alpha_{b,k} \cc_b^\conj(\nu_k) \otimes \bb(\tau_k)  [\arx(\theta_k)]_i \,,
\end{align}
has the following first and second moments:
\begin{align} \label{eq_hhib_moments}
    \hhbar_{i,b} \triangleq \Eee\{\hh_{i,b}\}, ~ \CC_{i,b} \triangleq \Eee\{(\hh_{i,b} -\hhbar_{i,b})(\hh_{i,b} -\hhbar_{i,b})^H\} ~.
\end{align}
Based on \eqref{eq_hhib_vec}, the vectorized observations in \eqref{eq_y_comp_beam} can be written as 
\begin{align} \label{eq_y_comp_vec}
    \yy_{i,b} = \boldD_{i,b} \hh_{i,b} + \nn_{i,b} \,,
\end{align}
where $\yy_{i,b} \triangleq \vecc{\boldY_{i,b}}$, $\boldD_{i,b} \triangleq \diag{\vecc{\boldX_b}}$ and $\nn_{i,b} \triangleq \vecc{\boldN_{i,b}}$. Given the statistics in \eqref{eq_hhib_moments}, the LMMSE estimate of the radar channel $\hh_{i,b}$ in \eqref{eq_y_comp_vec} is given by \cite[p.~389]{kay1993fundamentals}
\begin{align} \nonumber
    \hhhat_{i,b} &= \CC_{i,b} \boldD_{i,b}^H (\boldD_{i,b} \CC_{i,b} \boldD_{i,b}^H + \sigma^2 \Imatrix)^{-1} 
    \\ \label{eq_hhhat} & ~~~~~~ \times (\yy_{i,b} - \boldD_{i,b} \hhbar_{i,b}) + \hhbar_{i,b} \,.
\end{align}
The following lemma helps simplification of \eqref{eq_hhhat}.
\begin{lemma}[OFDM Radar Channel Statistics]\label{lemma_stat}
    Suppose that the target parameters in the radar sensing channel \eqref{eq_hhib_vec} are distributed independently as $\alpha_{b,k} \sim \mtCN(0, \sigmaakb^2), \tau_k \sim \mathcal{U}[0, 1/\deltaf], \nu_k \sim \mathcal{U}[0, 1/(\fc \Tsym)], \theta_k \sim \mathcal{U}[-\pi/2, \pi/2]  $ (i.e., delays, Dopplers and angles are drawn uniformly from their respective unambiguous detection intervals). In addition, different targets are assumed to have independent distributions.\footnote{In each OFDM frame, we sample a realization from these distributions, which we assume to stay constant during the frame.} Then, the statistics of $\hh_{i,b}$ in \eqref{eq_hhib_moments} are given by $ \hhbar_{i,b} = \boldzero $ and $\CC_{i,b} = \sigmaaab^2\Imatrix$, 
    where $\sigmaaab^2 \triangleq \sum_{k=0}^{K-1} \sigmaakb^2 $.
\end{lemma}
\begin{proof}
    See Appendix~\ref{app_lemma_stat}. 
\end{proof}
Based on Lemma~\ref{lemma_stat}, \eqref{eq_hhhat} becomes 
\begin{align} \label{eq_hhhat2}
   \hhhat_{i,b} &= \sigmaaab^2 \boldD_{i,b}^H (\sigmaaab^2 \boldD_{i,b} \boldD_{i,b}^H + \sigma^2  \Imatrix)^{-1} \yy_{i,b}  \,.
\end{align}
Folding \eqref{eq_hhhat2} back into matrix and plugging the definition of $\boldD_{i,b}$, the LMMSE channel estimate is given by
\begin{align} \label{eq_lmmse_b}
        \boldHhat_{i,b} = \frac{ \boldY_{i,b} \odot \boldX_b^\conj }{ \abs{\boldX_b}^2 + \snr_b^{-1}}  \,,
\end{align}
 where $\snr_b \triangleq \sigmaaab^2/\sigma^2$, and the $\abss{\cdot}^2$ operation is performed element-wise.

\subsubsection{Interpretation of LMMSE Estimator}\label{sec_int_lmmse}
A comparative analysis of \eqref{eq_lmmse_b} with \eqref{eq_rf} and \eqref{eq_mf} yields the following insights:
\begin{itemize}
    \item \textit{At high SNRs}, i.e., as $\snr_b \to \infty$, we have $\boldHhat_{i,b} \approx \boldY_{i,b} \oslash \boldX_b$. Thus, the LMMSE estimator converges to the RF estimator in \eqref{eq_rf}.
    \item \textit{At low SNRs}, i.e., as $\snr_b \to 0$, we obtain $\boldHhat_{i,b} \approx  \boldY_{i,b} \odot \boldX_b^\conj $. This suggests that the LMMSE estimator converges to the MF estimator in \eqref{eq_mf}. 
\end{itemize}
Hence, the LMMSE estimator represents a generalization of the RF and MF receivers to the entire SNR range. When using unit-amplitude modulations, where $\abss{\boldX_b}$ is an all-ones matrix, all estimators become equivalent, regardless of the SNR level. 

\subsubsection{Illustrative Example}
\begin{figure}
        \begin{center}
        \subfigure[]{
			 \label{fig_rangeProfile_twoTargets_lowSNR}
			 \includegraphics[width=0.9\columnwidth]{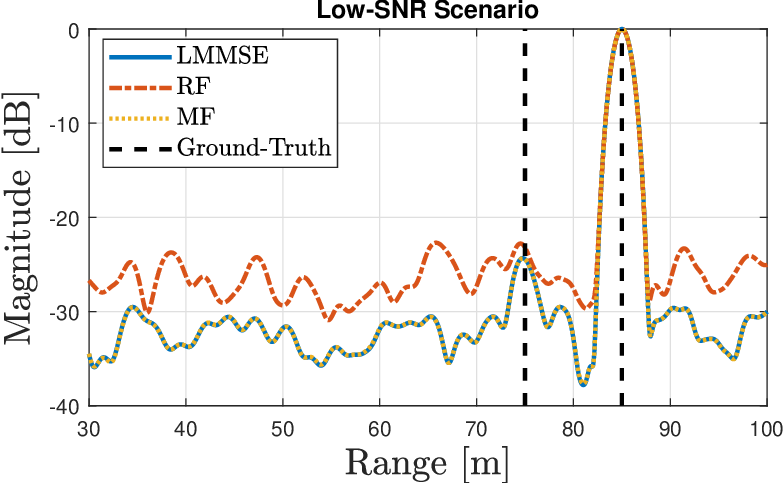}  
		}
        \subfigure[]{
			 \label{fig_rangeProfile_twoTargets_highSNR}
			 \includegraphics[width=0.9\columnwidth]{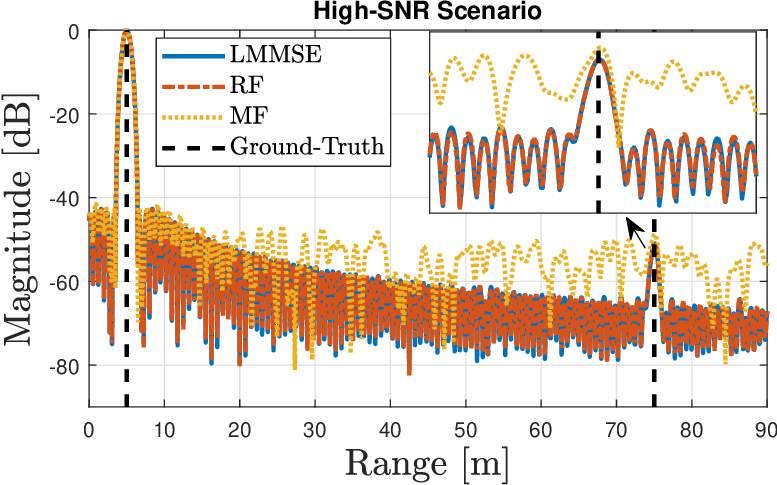}
		}		
		\end{center}
		\vspace{-0.1in}
        \caption{Comparative analysis of range profiles obtained by different channel estimation strategies with $1024-$QAM modulation. \subref{fig_rangeProfile_twoTargets_lowSNR} Low-SNR scenario with the OFDM parameters in Table~\ref{tab_parameters} except $B = 100 \, \rm{MHz}$ and $M = 140$. The scenario involves two targets, each moving at $15 \, \rm{m/s}$, with RCSs of $(-7, 20) \, \rm{dBsm}$ and located at ranges of $(75, 85) \, \rm{m}$ and angles of $(10^\circ, 10^\circ)$. \subref{fig_rangeProfile_twoTargets_highSNR} High-SNR scenario with the OFDM parameters in Table~\ref{tab_parameters} except $B = 200 \, \rm{MHz}$ and $M = 140$. The scenario involves two targets, each moving at $15 \, \rm{m/s}$, with RCSs of $(15, 20) \, \rm{dBsm}$ and located at ranges of $(75, 5) \, \rm{m}$ and angles of $(10^\circ, 10^\circ)$. As opposed to the RF and MF estimators, which fail, respectively, at low and high SNRs, the proposed LMMSE approach provides a globally effective solution that achieves low side-lobe levels under high-order QAM modulations across diverse SNR conditions.}  
        \label{fig_rangeProfile_CE_Comp}
\end{figure}

We provide an illustrative example in Fig.~\ref{fig_rangeProfile_CE_Comp} to showcase the behavior of the estimators in \eqref{eq_rf}, \eqref{eq_mf} and \eqref{eq_lmmse_b} with $1024-$QAM modulation. The figure demonstrates the range profiles obtained by 2-D FFT \cite{RadCom_Proc_IEEE_2011,OFDM_Radar_Phd_2014,OFDM_Radar_Corr_TAES_2020,MIMO_OFDM_ICI_JSTSP_2021} on $\boldHhat_{i,b}$ under two different scenarios corresponding to low and high SNR operation.\footnote{\label{fn_snr_mean}Defining $\snrmean \triangleq \sum_{b=1}^{B} \snr_b /B$, the low and high SNR scenarios have $\snrmean = -29.7 \, \rm{dB}$ and $\snrmean = 19.5 \, \rm{dB}$, respectively.} We observe that the results corroborate the interpretations in Sec.~\ref{sec_int_lmmse}, i.e., the LMMSE estimator converges to the MF and the RF estimators at low and high SNR regimes, respectively. As seen from Fig.~\ref{fig_rangeProfile_CE_Comp}, the proposed LMMSE approach emerges as a robust and globally effective solution, maintaining low side-lobe levels under high-order QAM modulations across a range of SNR conditions. In contrast, conventional RF and MF estimators struggle with the detection of weak targets due to increased side-lobe levels at low and high SNRs, respectively. Drawing from these findings, Table~\ref{tab_guideline} offers general guidelines for selecting estimators to achieve robust sensing performance across various SNR levels.


\begin{table}
\caption{Guidelines on Which Estimator to Use In OFDM Radar Sensing with High-Order Modulations Under Different SNR Regimes.}
\centering
    \begin{tabular}{|c|c|c|c|}
        \hline
        & {RF} & {MF} & {LMMSE} \\ \hline
        {Low SNR}  & \ding{55} & \ding{51} & \ding{51} \\ \hline
        {High SNR} & \ding{51} & \ding{55} & \ding{51} \\ \hline
    \end{tabular}
    \label{tab_guideline}
  \vspace{0.1in}
\end{table}

\vspace{-4mm}
\subsection{Beam-Specific Delay-Doppler Estimation}\label{sec_dd_est}
We are now interested in estimating target delays and Dopplers from the output of channel estimation $\boldHhat_{i,b}$ in \eqref{eq_rf}, \eqref{eq_mf} or \eqref{eq_lmmse_b}. The channel estimates will have the form
\begin{align} \label{eq_hhat_ib}
    \boldHhat_{i,b} = \boldH_{i,b} + {\boldZ}_{i,b} \in \complexset{N}{M_b} \,,
\end{align}
where $\boldH_{i,b}$ is given in \eqref{eq_hib} and ${\boldZ}_{i,b}$ represents the channel estimation error term. Based on \eqref{eq_hib} and following a generalized likelihood ratio test (GLRT) approach \cite{MIMO_OFDM_ICI_JSTSP_2021}, the delay-Doppler images can be computed simply via 2-D FFT, i.e.,
\begin{align}\label{eq_2d_fft_step3}
     h_{i,b}(\tau, \nu) = \bb^H(\tau) \boldHhat_{i,b} \cc_b(\nu) \,,
\end{align}
where $\bb(\tau)$ in \eqref{eq_steer_delay} and $\cc_b(\nu)$ in \eqref{eq_steer_doppler} coincide with (zero-padded) DFT matrix columns. In matrix form with sampled delay-Doppler grid, \eqref{eq_2d_fft_step3} can be written as
\begin{align} \label{eq_hhatdd}
    \hhatdd_{i,b} =  \FF_N^H \boldHhat_{i,b} \FF_{M_b}  \in \complexset{N}{M_b} \,,
\end{align}
where $\FF_N \in \complexset{N}{N}$ is the unitary DFT matrix of size $N$. Since target angles are a-priori unknown and 3-D delay-Doppler-angle processing is computationally demanding, we propose to non-coherently integrate over the antenna elements \cite{MIMO_OFDM_radar_TAES_2020,MIMO_OTFS_ISAC_2023} to obtain the delay-Doppler image for each beam $b$ (i.e., without seeking phase-alignment with $[\arx(\theta_k)]_i$ in \eqref{eq_hib}):
\begin{align} \label{eq_hhatdd_noncoh}
    \hhatdd_{b} = \sum_{i=1}^{\Nrx}  \abs{\hhatdd_{i,b}}^2 \,.
\end{align}
Multiple target detection and delay-Doppler estimation can be performed by searching for peaks in \eqref{eq_hhatdd_noncoh}, e.g., via a constant false alarm rate (CFAR) detector \cite[Ch.~6.2.4]{richards2005fundamentals}.

\vspace{-3mm}
\subsection{Multi-Target Angle Estimation per Delay-Doppler Cell}
\label{sec_ang_est}
For each delay-Doppler detection $(\tauhat, \nuhat)$ obtained via 2-D CFAR on $\hhatdd_b$ in \eqref{eq_hhatdd_noncoh}, we estimate angles of (possibly) multiple targets residing within the resolution cell of $(\tauhat, \nuhat)$. To this end, we construct the spatial domain compressed observation $\yyhat_b \in \complexset{\Nrx}{1}$ from the original channel estimates in \eqref{eq_hhat_ib} as follows: 
\begin{align}\nonumber
    [\yyhat_b]_i &= \frac{1}{N M_b} \bb^H(\tauhat) \boldHhat_{i,b} \cc_b(\nuhat) \,,
    \\ \nonumber
    &= \sum_{k=0}^{K-1} \alpha_{b,k} \frac{\bb^H(\tauhat) \bb(\tau_k)}{N} \frac{\cc_b^H(\nu_k) \cc_b(\nuhat)}{M_b} [\arx(\theta_k)]_i + [\nn_b]_i \,,
    \\ \label{eq_ybi}
    & 
    \approx  \sum_{k \in \kthh} \alpha_{b,k} [\arx(\theta_k)]_i  + [\nn_b]_i \,,
\end{align}
where $[\nn_b]_i \triangleq \frac{1}{N M_b} \bb^H(\tauhat) {\boldZ}_{i,b} \cc_b(\nuhat) $ is the noise component, $\kthh \subseteq \{1, \ldots, K\}$ contains the indices of the targets located within the resolution cell of $(\tauhat, \nuhat)$ (i.e., $\abss{\bb^H(\tauhat) \bb(\tau_{k_1})} \gg \abss{\bb^H(\tauhat) \bb(\tau_{k_2})}$ and $\abss{\cc_b^H(\nuhat) \cc_b(\nu_{k_1})} \gg \abss{\cc_b^H(\nuhat) \cc_b(\nu_{k_2})}$ for $k_1 \in \kthh, k_2 \notin \kthh$). 
Stacking \eqref{eq_ybi} over the RX array yields
\begin{align} \label{eq_yy_angle}
    \yyhat_b =  \sum_{k \in \kthh} \alpha_{b,k} \arx(\theta_k) + \nn_b \,.
\end{align}

Due to relatively small number of antenna elements $\Nrx$ at typical ISAC sensing receivers, 
we propose to retrieve multiple angles from \eqref{eq_yy_angle} using the ESPRIT algorithm in order to resolve targets closely spaced in the angular domain. In particular, we resort to 1-D ESPRIT using spatial smoothing with Hankel matrix construction \cite{delayTarik_TWC_2022,ge2023v2x}. Due to space limitations, the reader is referred to \cite[Sec.~IV-D]{delayTarik_TWC_2022} and \cite[Sec.~III-A]{ge2023v2x} for details on multiple angle estimation using \eqref{eq_yy_angle}.

\vspace{-3mm}
\subsection{Beam-Specific Gain Estimation}\label{sec_gain_est}
Given the estimates $\mest_b = \{\tauhatbi{b}{k}, \nuhatbi{b}{k}, \thetahatbi{b}{k}\}_{k=1}^{K_b}$ from the beam-specific processing for the $\thn{b}$ beam, we estimate the corresponding target gains $\alphab_b \in \complexset{K_b}{1}$ via least-squares (LS) using the channel estimates $\{\boldHhat_{i,b}\}_{i=1}^{\Nrx}$ in \eqref{eq_hhat_ib}. Based on the channel structure in \eqref{eq_hhat_ib} and \eqref{eq_hib}, the channel estimates over all RX antenna elements can be expressed as
\begin{align} \label{eq_hhb_2}
    \hhhat_b \triangleq \begin{bmatrix}
        \hhhat_{1,b} \\ \vdots \\ \hhhat_{\Nrx,b} 
    \end{bmatrix} = \underbrace{\begin{bmatrix}
        \boldA_{1,b} \\ \vdots \\ \boldA_{\Nrx,b} 
    \end{bmatrix}}_{\triangleq \boldA_b \in \complexset{N M_b \Nrx}{K_b} } \alphab_b + \zz_b ~,
\end{align}
where $\hhhat_{i,b} \triangleq \veccs{\boldHhat_{i,b}} \in \complexset{N M_b}{1}$,  $ \boldA_{i,b} \triangleq [ \aaa_{i,b}^{(1)} ~ \ldots ~ \aaa_{i,b}^{(K_b)} ] \in \complexset{N M_b}{K_b} ~,$, and $ \aaa_{i,b}^{(k)} \triangleq [\arx(\thetahatbi{b}{k})]_i \,\cc_b^\conj(\nuhatbi{b}{k}) \otimes \bb(\tauhatbi{b}{k}) \in \complexset{N M_b}{1} $. 
Using \eqref{eq_hhb_2}, the gain estimates can be obtained as 
\begin{align} \label{eq_alphahatb}
    \alphahatb_b = \boldA_b^{\dagger} \hhhat_b ~. 
\end{align}

\vspace{-4mm}
\subsection{LMMSE Processing with A-Priori Unknown 
Gains}
The LMMSE estimator expression in \eqref{eq_lmmse_b} involves the SNR term $\snr_b$, the evaluation of which requires the knowledge of the target channel gains (i.e., $\sigmaakb^2 = \Eee\{ \lvert \alpha_{b,k} \rvert^2 \} $). Since the channel gains are a-priori unknown, we propose to first run the LMMSE estimator in \eqref{eq_lmmse_b} using three different $\snr_b$ values, following the subsequent processing steps in Sec.~\ref{sec_dd_est}--Sec.~\ref{sec_ang_est}, and merge and cluster the detections obtained via the different $\snr_b$ values (see Lines~\ref{line_lmmse}--\ref{line_merge} of Algorithm~\ref{alg_overall}). Then, we plug the resulting gain estimates from \eqref{eq_alphahatb} into the SNR expression \eqref{eq_lmmse_b}, i.e.,
\begin{align}\label{eq_snr_gain}
    \snr_b = \norm{ \alphahatb_b}^2/\sigma^2 \,,
\end{align}
which can now be inserted into \eqref{eq_lmmse_b} for LMMSE estimation.

\vspace{-4mm}
\subsection{Clustering Detections over All Beams}\label{sec_cluster}
\vspace{-1mm}
Since the same target might be detected in multiple beams during beam sweeping, we need a clustering algorithm to merge detections from all beams. To this end, given the estimates $\{\mest_b\}_{b=1}^{B}$, we resort to the density-based spatial clustering of applications with noise (DBSCAN) algorithm \cite{dbscan,dbscan_ComMag_2018} to cluster detections over $B$ beams. In DBSCAN, we set the minimum number of points in a cluster to $1$ and define the distance measure (used to characterize the $\epsilon$-neighbourhood of a point \cite{dbscan}) in the range-velocity-angle space as the weighted Euclidean distance:
\begin{align}
    d_{\boldW}(\sssh_p, \sssh_q) = [ (\sssh_p - \sssh_q)^T \boldW  (\sssh_p - \sssh_q) ]^{1/2} \,, 
\end{align}
where $\sssh_p = [\Rhat_p, \vhat_p, \thetahat_p]^T$ contains the range, velocity and angle estimates of the $\thn{p}$ detection in $\{\mest_b\}_{b=1}^{B}$, and $\boldW \in \realset{3}{3}$ is a diagonal matrix to account for scaling due to unit differences. The $\epsilon$ parameter \cite{dbscan} is set to $\epsilon = ( \sdelta^T \boldW \sdelta )^{1/2}$, where $\sdelta = [\Delta R, \Delta v, \Delta \theta ]^T$ with $\Delta R = c/(2N\deltaf)$ denoting the range resolution, $\Delta v = \lambda B / (2 M \Tsym)$ the velocity resolution per beam and\footnote{Since the ESPRIT algorithm in Sec.~\ref{sec_ang_est} offers higher angular resolution than the standard Rayleigh resolution of $\arx(\theta)$ in \eqref{eq_ar_steer}, we set the proximity criterion in angle to a lower value than the standard resolution.} $\Delta \theta = 2^\circ$.

\vspace{-3mm}
\subsection{Summary of the Proposed Algorithm}
The proposed LMMSE based sensing algorithm is summarized in Algorithm~\ref{alg_overall}, which uses Algorithm~\ref{alg_beam_spec} as a subroutine.

\begin{algorithm}[t]
	\caption{Beam-Specific Parameter Estimation from Unstructured Radar Channel Estimates}
	\label{alg_beam_spec}
	\begin{algorithmic}[1]
		\State \textbf{Input:} Frequency/slow-time radar channel estimates $\{\boldHhat_{i,b} \}_{i=1}^{\Nrx}$ for the $\thn{b}$ beam in \eqref{eq_rf}, \eqref{eq_mf} or \eqref{eq_lmmse_b}, probability of false alarm $\pfa$.
		\State \textbf{Output:} Delay, Doppler, angle and gain estimates $\{\tauhatbi{b}{k}, \nuhatbi{b}{k}, \thetahatbi{b}{k}, \alphahatbi{b}{k} \}_{k=1}^{K_b}$ of multiple targets.
        \State \parbox[t]{230pt}{Obtain the noncoherently integrated delay-Doppler image $\hhatdd_{b}$ for the $\thn{b}$ beam via \eqref{eq_hhatdd} and \eqref{eq_hhatdd_noncoh}.\strut}
        \State \parbox[t]{230pt}{Run a CFAR detector on $\hhatdd_{b}$ with the specified $\pfa$ for target detection in the delay-Doppler domain.\strut}
        \State  \parbox[t]{230pt}{For each delay-Doppler detection $(\tauhat, \nuhat)$, compute the spatial domain observation $\yyhat_b$ in \eqref{eq_ybi}.\strut}
        \State Estimate angles from $\yyhat_b$ via 1-D ESPRIT.
        \State Using the estimates $\{\boldHhat_{i,b}\}_{i=1}^{\Nrx}$ and $\{\tauhatbi{b}{k}, \nuhatbi{b}{k}, \thetahatbi{b}{k}\}_{k=1}^{K_b}$, estimate the gains via \eqref{eq_hhb_2}--\eqref{eq_alphahatb}.
	\end{algorithmic} 
	\normalsize
\end{algorithm}

\begin{algorithm}[t]
	\caption{LMMSE-Based MIMO-OFDM Radar Sensing}
	\label{alg_overall}
	\begin{algorithmic}[1]
		\State \textbf{Input:} Space/frequency/slow-time MIMO-OFDM radar data cube $\{ \yy_{n,m} \}$ in \eqref{eq_ynm}, transmit symbols $\boldX$, number of distinct beams $B$, probability of false alarm $\pfa$.
		\State \textbf{Output:} Delay-Doppler-angle estimates $\{\tauhat_k, \nuhat_k, \thetahat_k\}_{k=0}^{K-1}$ of multiple targets.
		\State \textbf{for} $b = 1, \ldots, B$
        \Indent
        \State \parbox[t]{210pt}{Compute the RF, MF and LMMSE channel estimate $\boldHhat_{i,b}$ in \eqref{eq_rf}, \eqref{eq_mf} and \eqref{eq_lmmse_b}, respectively, for each RX element $i = 1, \ldots, \Nrx$ using $\snr_b = 1$ in \eqref{eq_lmmse_b}.\strut} \label{line_lmmse}
        \State \parbox[t]{210pt}{For each estimate, run \textbf{Algorithm~\ref{alg_beam_spec}} to obtain the delay, Doppler and angle detections $\mest_b = \{\tauhatbi{b}{k}, \nuhatbi{b}{k}, \thetahatbi{b}{k}\}_{k=1}^{K_b}$.\strut} \label{line_gain}   
        \vspace{0.05in}
        \State \parbox[t]{210pt}{Merge the detections from the three estimators and cluster them via DBSCAN in Sec.~\ref{sec_cluster}.}\label{line_merge}
        \vspace{0.05in}
         \State \parbox[t]{210pt}{Compute the resulting $\snr_b$ via \eqref{eq_alphahatb} and \eqref{eq_snr_gain}.}\label{line_snr_b}
         \State \parbox[t]{210pt}{Compute the LMMSE channel estimate in \eqref{eq_lmmse_b} by inserting $\snr_b$ and run \textbf{Algorithm~\ref{alg_beam_spec}}.\strut}\label{line_lmmse_final} 
        \EndIndent
        \State \textbf{end for} 
        \State Perform DBSCAN clustering of the resulting detections $\{\mest_b\}_{b=1}^{B}$ using the settings in Sec.~\ref{sec_cluster}.
	\end{algorithmic} 
	\normalsize
\end{algorithm}

\vspace{-2mm}
\section{Communications Data Rate Evaluation}\label{sec_comm_rate}
In this section, we provide a methodology to evaluate the data rate of the communications subsystem of the considered ISAC system under different constellations, 
similar to \cite{Liu_Reshaping_OFDM_ISAC_2023}.
\vspace{-3mm}
\subsection{Calculation of Mutual Information}\label{sec_calc_MI}
The received signal at the communications RX in \eqref{eq_ynmcom} can be recast as $\ycom_{n,m} = \hnm \xnm + z_{n,m}$, 
where $\hnm \triangleq (\hhcom_{n,m})^T \ff_m$.
As we accumulate the values of $\ycom_{n,m}$ over the entire OFDM frame, we obtain
\begin{align} \label{eq_Ycom}
\boldYcom = \mathbf{H}\odot\boldX+\boldZ \,,
\end{align}
where $\boldYcom\in\mathbb{C}^{N\times M}$ with $[\boldYcom]_{n,m}=\ycom_{n,m}$, the effective channel matrix $\boldH\in\mathbb{C}^{N\times M}$ with $[\mathbf{H}]_{n,m}=\hnm$, and $\boldZ\in\mathbb{C}^{N\times M}$ is a noise matrix with \ac{i.i.d.} entries $z_{n,m} \sim \mtCN(0, \sigmacsq )$. Hence, assuming perfect knowledge of the channel at the RX, the \ac{MI} of the entire frame can be written as
\begin{align} \label{eq_mi_ycom}
I(\mathbf{X};\boldYcom|\mathbf{H})=\sum_{n,m} I(\xnm;\ycom_{n,m}\,|\,\hnm) \,,
\end{align}
where the \ac{MI} term is given by
\begin{align} \label{eq_mi}
    I(\xnm;\ycom_{n,m}\,|\,\hnm)&=h(\ycom_{n,m}\,|\,\hnm)-h(\ycom_{n,m}\,|\,\hnm,\xnm) \,,
\end{align}
with $h(\cdot)$ representing the (conditional) entropy of a random variable. The separability of the \ac{MI} in \eqref{eq_mi_ycom} results from the independence of $\ycom_{n,m}$ across subcarriers and symbols.
\vspace{-3mm}
\subsection{Evaluation of Entropy}\label{sec_calc_entr}
The entropy on the right-hand side of \eqref{eq_mi} is given by \cite{thomas2006elements}
\begin{align} \label{eq_entropy_gauss}
    & h(\ycom_{n,m}\,|\,\hnm,\xnm) 
    \\ \nonumber &=  \int_{\mathbb{C}}-f(\ycom_{n,m}\,|\,\hnm,\xnm)\log f(\ycom_{n,m}\,|\,\hnm,\xnm)\, \d \ycom_{n,m}  \,,
    \\ \nonumber &= \log\left(\pi e \sigmacsq\right) \,,
\end{align}
which is due to $f\!\left(\ycom_{n,m}\,|\,\hnm,\xnm\right)=\mtCN(\hnm\xnm,\sigmacsq)$. We now calculate $h(\ycom_{n,m}\,|\,\hnm)$ in \eqref{eq_mi}  via
\begin{align} \nonumber
    h(\ycom_{n,m}\,|\,\hnm)   
    &=\mathbb{E}\left[-\log f(\ycom_{n,m}\,|\,\hnm)\right] \,,
    \\ 
    &\approx -\frac{1}{\ns}\sum_{i=1}^{\ns} \log f(y_i|\,\hnm) \,, \label{eqn_approx_LLN}
\end{align}
where the approximation leverages the law of large numbers and becomes exact as $\ns$ increases. Assuming that $\xnm$ is uniformly distributed over a finite alphabet $\mathcal{X}=\{x_1,x_2,\ldots,x_L\}$, \eqref{eqn_approx_LLN} can be evaluated via 
\begin{align}
\label{eq_output_prob}
    f(\ycom_{n,m}\,|\,\hnm)&=\frac{1}{L}\sum_{\ell=1}^{L}f(\ycom_{n,m}\,|\,\hnm,x_\ell)\\ \nonumber
    &=\frac{1}{L}\sum_{\ell=1}^{L} \frac{1}{\pi\sigmacsq}\exp\Big(-\frac{|\ycom_{n,m}-\hnm x_\ell|^2}{\sigmacsq}\Big).
\end{align}
\vspace{-3mm}
\subsection{General Procedure to Evaluate MI in \eqref{eq_mi_ycom}}
Drawing from Sec.~\ref{sec_calc_MI} and Sec.~\ref{sec_calc_entr}, we outline the MI evaluation process from \eqref{eq_mi_ycom} in Algorithm~\ref{alg_rate}.

\begin{algorithm}[t]
	\caption{Rate Evaluation under Different Constellations}
	\label{alg_rate}
	\begin{algorithmic}[1]
		\State \textbf{Input:} Frequency/slow-time communication channel $\boldH $ in \eqref{eq_Ycom}, alphabet $\Xcal=\{x_1,x_2,\ldots,x_L\}$, number of samples for entropy approximation $\ns$.
		\State \textbf{Output:} MI in \eqref{eq_mi_ycom}. 
        \State \textbf{for} $n = 0, \ldots, N-1$, $m = 0, \ldots, M-1$
        \Indent
        \State \textbf{for} $i = 1, \ldots, \ns$
        \Indent
        \State \parbox[t]{200pt}{Select a symbol $x_i$ from the alphabet $\Xcal$  in a uniformly random manner.\strut} 
        \State \parbox[t]{200pt}{Generate a realization of $n_{n,m} \sim \mtCN(0,\sigmacsq)$.\strut}
        \State Calculate the output via $y_i=\hnm x_i+n_{n,m}$.
        \State Compute $f(y_i|\,\hnm)$ via \eqref{eq_output_prob}.
        \EndIndent
        \State \textbf{end for} 
        \State Compute $ h(\ycom_{n,m}\,|\,\hnm) $ via \eqref{eqn_approx_LLN}. \label{line_hycom}
        \State Compute the MI in \eqref{eq_mi} via \eqref{eq_entropy_gauss} and \eqref{eqn_approx_LLN}. \label{line_mi}
        \EndIndent
        \State \textbf{end for}
        \State Compute the MI of the entire frame in \eqref{eq_mi_ycom} via \eqref{eq_mi}.
	\end{algorithmic} 
	\normalsize
\end{algorithm}

\section{Simulation and Experimental Results}\label{sec_num}
In this section, we assess the performance of the proposed sensing algorithm in Algorithm~\ref{alg_overall} as well as the accompanying ISAC trade-offs using the data rate evaluation procedure in Algorithm~\ref{alg_rate} on both simulated and experimentally obtained data. The default simulation parameters are provided in Table~\ref{tab_parameters}. For sensing beam sweeping in Sec.~\ref{sec_bf_senscom}, we set the number of beams $B = \Ntx$ and $\thetamax = 70^\circ$. To evaluate the detection performance, we set $\pfa = 10^{-4}$ in Algorithm~\ref{alg_overall} and consider $100$ independent Monte Carlo noise realizations in \eqref{eq_ynm}, where each realization corresponds to a 3-D noise tensor of size $\Nrx N M$. For benchmarking purposes, we compare the sensing and ISAC performance of the following algorithms:
\begin{itemize}
    \item \textbf{LMMSE}: The proposed MIMO-OFDM radar sensing algorithm in Algorithm~\ref{alg_overall}. 
    \item \textbf{LMMSE (ideal)}: The genie-aided version of Algorithm~\ref{alg_overall} where the true value of $\snr_b$ is inserted into \eqref{eq_lmmse_b} on Line~\ref{line_lmmse_final} by skipping Lines~\ref{line_lmmse}--\ref{line_snr_b}. This serves as an upper bound on the performance of LMMSE. 
    \item \textbf{RF}: The standard RF-based sensing algorithm that executes Algorithm~\ref{alg_overall} by skipping Lines~\ref{line_merge}--\ref{line_lmmse_final}.
    \item \textbf{MF}: The standard MF-based sensing algorithm that executes Algorithm~\ref{alg_overall} by skipping Lines~\ref{line_merge}--\ref{line_lmmse_final}.
    \item \textbf{RF-SOTA}: The state-of-the-art RF-based MIMO-OFDM radar sensing algorithm in \cite{5G_NR_JRC_analysis_JSAC_2022}.
\end{itemize}

For the communication channel in \eqref{eq_hhcom}, we consider the presence of $\Ktilde = 4$ paths to an RX located at $[43, -25]^T \, \rm{m}$ with respect to the ISAC TX, including a \ac{LOS} path and $3$ \ac{NLOS} paths. The corresponding scatterers are located at $[40, -20]^T \, \rm{m}$, $[42, -27]^T \, \rm{m}$ and $[38, -30]^T \, \rm{m}$, with the RCSs $\{-5, -10, -10\} \, \rm{dBsm}$, respectively. Considering stationary RX and scatterers, we set $\nut_k = 0$ for all the paths.

In the following subsections, we first explore the time-frequency and spatial domain ISAC trade-offs under concurrent transmission. Then, we compare concurrent and time-sharing strategies as described in Sec.~\ref{sec_isac_str}. Finally, we provide experimental results using real-world OFDM monostatic sensing measurements to verify the behaviour of the considered algorithms as illustrated in Fig.~\ref{fig_range_profile_RF_QPSK_QAM} and Fig.~\ref{fig_rangeProfile_CE_Comp}.

\begin{table}
\caption{Simulation Parameters}
\centering
    \begin{tabular}{|l|l|}
        \hline
        Carrier frequency, $\fc$ & $28 \, \rm{GHz}$ \\ \hline
        Subcarrier spacing, $\deltaf$ & $120 \, \rm{kHz}$  \\ \hline
        Number of subcarriers, $N$ & $3330$  \\ \hline
        Bandwidth, $B$ & $400 \, \rm{MHz}$    \\ \hline
        Number of symbols, $M$ & $1120$   \\ \hline
        Transmit power, $\Pt$ & $20 \, \rm{dBm}$  \\ \hline
        ULA size at the ISAC TX, $\Ntx$ & $8$   
        \\ \hline
        ULA size at the sensing RX, $\Nrx$  &  $8$
        \\ \hline
        Noise PSD & $-174 \, \rm{dBm/Hz}$  
         \\ \hline
    \end{tabular}
    \label{tab_parameters}
\end{table}
\vspace{-3mm}
\subsection{Time-Frequency Trade-offs under Concurrent Transmission}\label{sec_tf_tradeoff}
To investigate time-frequency domain ISAC trade-offs, the performance of the algorithms under concurrent transmission, as outlined in Sec.~\ref{sec_conc}, is evaluated using different modulation orders for transmit symbols $\boldX$ in \eqref{eq_ynm}. We set $\rho = 0.8$ in \eqref{eq_ffm} and consider two sensing scenarios, \textit{low-SNR} and \textit{high-SNR} (quantified by $\snrmean$ defined in Footnote~\ref{fn_snr_mean}), which are crucial for highlighting the distinct characteristics of the algorithms, as extensively demonstrated earlier in Fig.~\ref{fig_rangeProfile_CE_Comp}.

\subsubsection{Low-SNR Sensing Scenario}
In the low-SNR scenario, we consider a single target with range $80 \, \rm{m}$, velocity $15 \, \rm{m/s}$, angle $10^\circ$ and RCS $-2 \, \rm{dBsm}$, leading to $\snrmean = -50.6 \, \rm{dB}$. The goal is to investigate whether the target is drowned out by its own interference under high-order modulations due to increased side-lobe levels as depicted in Fig.~\ref{fig_range_profile_RF_QPSK_QAM} and Fig.~\ref{fig_rangeProfile_CE_Comp}. In Fig.~\ref{fig_pd_vs_rcs_singleTarget_Concurrent_1024_QAM}, we show the probability of detection ($\pd$) of various algorithms relative to target RCS, employing $1024-$QAM modulation for $\boldX$. It is observed that the proposed LMMSE estimator significantly outperforms the RF estimator, achieving gains in $\pd$ as high as $1$ for a constant RCS and enabling the detection of targets with RCS values smaller by up to $6 \, \rm{dBsm}$ for a specified $\pd$. Furthermore, consistent with the discussions in Sec.~\ref{sec_int_lmmse} and findings in Fig.~\ref{fig_rangeProfile_twoTargets_lowSNR}, LMMSE and MF achieve the same performances. Thus, in compliance with Table~\ref{tab_guideline}, LMMSE and MF are suitable for low-SNR sensing, whereas RF leads to severe degradations in probability of detection. Moreover, LMMSE matches the detection performance of its genie-aided version, which assumes perfect knowledge of $\snr_b$ in \eqref{eq_lmmse_b}, validating the effectiveness of the SNR estimation approach in Algorithm~\ref{alg_beam_spec}. 
Finally, when comparing our proposed MIMO-OFDM radar sensing approach to the state-of-the-art, the RF method in Algorithm~\ref{alg_overall} substantially outperforms the one in \cite{5G_NR_JRC_analysis_JSAC_2022}. Indeed, the latter method is more suitable in a tracking scenario where, if the target angular sector is known, the TX and RX beams can be steered towards the target, thereby exploiting receive combining gain. Conversely, in the considered scanning/search scenario where the target angular sector is unknown, the method in \cite{5G_NR_JRC_analysis_JSAC_2022} fails to fully exploit the available measurements at RX antennas, as reflected in the RCS loss shown in Fig.~\ref{fig_pd_vs_rcs_singleTarget_Concurrent_1024_QAM}.

To investigate ISAC trade-offs, we explore the impact of modulation order on the sensing and communication performances. To this end, in Fig.~\ref{fig_pd_vs_modOrder_singleTarget_Concurrent}  we depict the probability of detection and achievable rate across various modulations, from QPSK to $1024-$QAM, for a constant RCS of $-2 \, \rm{dBsm}$. The detection curves reveal the robustness of the proposed LMMSE approach against increasing modulation order, showcasing its capability to mitigate target masking effects under high-order modulations. Conversely, the RF estimator suffers from substantial loss in probability of detection with increasing modulation order due to rising side-lobe levels, corroborating the findings in Fig.~\ref{fig_range_profile_RF_QPSK_QAM}. Furthermore, the achievable rate improves as modulation order increases, as expected, which suggests that LMMSE (along with its genie-aided, ideal version and MF) achieves much more favorable ISAC trade-offs compared to RF. Hence, employing LMMSE in the sensing receiver enables significant improvements in communication rates through high-order QAM signaling without compromising detection performance.

\begin{figure}[t]
        \begin{center}
        \subfigure[]{
			 \label{fig_pd_vs_rcs_singleTarget_Concurrent_1024_QAM}
			 \includegraphics[width=0.9\columnwidth]{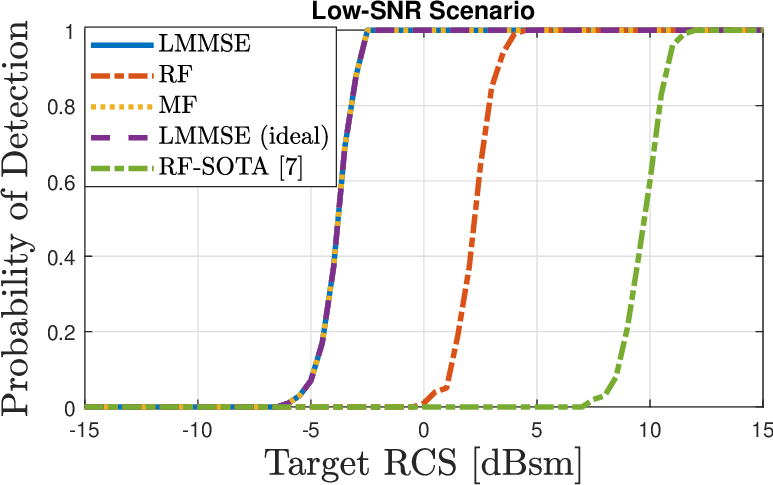}
		}
        \subfigure[]{
			 \label{fig_pd_vs_modOrder_singleTarget_Concurrent}
			 \includegraphics[width=0.9\columnwidth]{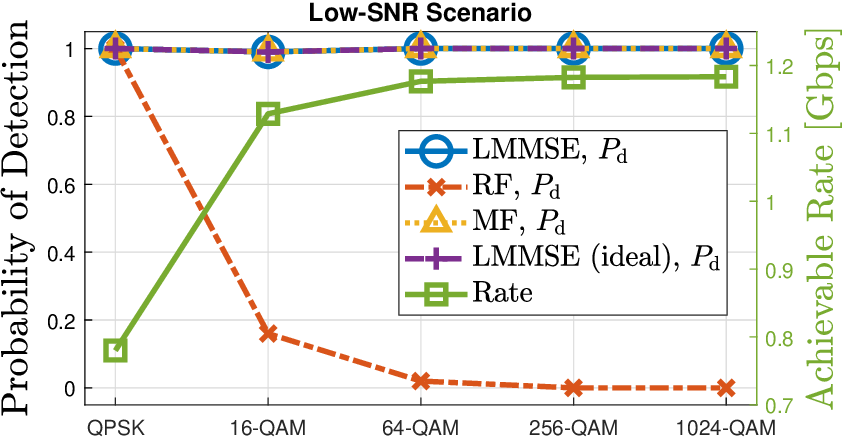}
		}
		
		\end{center}
		\vspace{-0.1in}
        \caption{ISAC performances in the low-SNR sensing scenario under concurrent transmission. \subref{fig_pd_vs_rcs_singleTarget_Concurrent_1024_QAM} Probability of detection with respect to target RCS using $1024-$QAM data. \subref{fig_pd_vs_modOrder_singleTarget_Concurrent} Probability of detection and achievable rate with respect to modulation order for target RCS of $-2 \, \rm{dBsm}$.}  
        \label{fig_lowSNR_TF_tradeoff}
\end{figure}

\subsubsection{High-SNR Sensing Scenario}\label{sec_num_highSNR_TF}
We now turn our attention to a high-SNR scenario which contains three targets with ranges of $(75, 5, 5) \, \rm{m}$, velocities of $(15, -10, -10) \, \rm{m/s}$, angles of $(10^\circ, 10^\circ, 18^\circ)$ and RCSs of $(5, 20, 5) \, \rm{dBsm}$, resulting in $\snrmean = 19.6 \, \rm{dB}$. In Fig.~\ref{fig_pd_vs_rcs_threeTargets_Concurrent_1024_QAM}, we show the probability of detection of Target-$1$ relative to its RCS with $1024-$QAM modulation. It is seen that the proposed LMMSE estimator outperforms both RF and MF benchmarks, indicating that it can effectively reduce side-lobe levels to prevent target masking when using varying-amplitude data. In particular, LMMSE provides up to $11 \, \rm{dBsm}$ gain in RCS for fixed $\pd$ over MF. In accordance with Table~\ref{tab_guideline} and Fig.~\ref{fig_rangeProfile_twoTargets_highSNR}, at high sensing SNRs, MF experiences significant loss in detection performance, while RF leads to slight degradations in performance compared to LMMSE. Moreover, in Fig.~\ref{fig_pd_vs_modOrder_threeTargets_Concurrent_highSNR} we report the probability of detection of Target-$1$ and achievable rate against modulation order for a fixed Target-$1$ RCS of $5 \, \rm{dBsm}$. We observe that MF fails to detect the target when employing $16$-QAM modulation and higher, while LMMSE maintains a constant $\pd$ of $1$ across all modulation orders. 

\subsubsection{Summary of Low-SNR and High-SNR Scenarios}
Fig.~\ref{fig_lowSNR_TF_tradeoff} and Fig.~\ref{fig_highSNR_TF_tradeoff} reveal that LMMSE provides consistently superior detection performance across different SNR regimes over existing benchmarks when using high-order QAM signaling. This establishes LMMSE as a universally effective strategy for OFDM radar sensing across a spectrum of SNR levels and modulation schemes.

\begin{figure}[t]
        \begin{center}
        \subfigure[]{
			 \label{fig_pd_vs_rcs_threeTargets_Concurrent_1024_QAM}
			 \includegraphics[width=0.9\columnwidth]{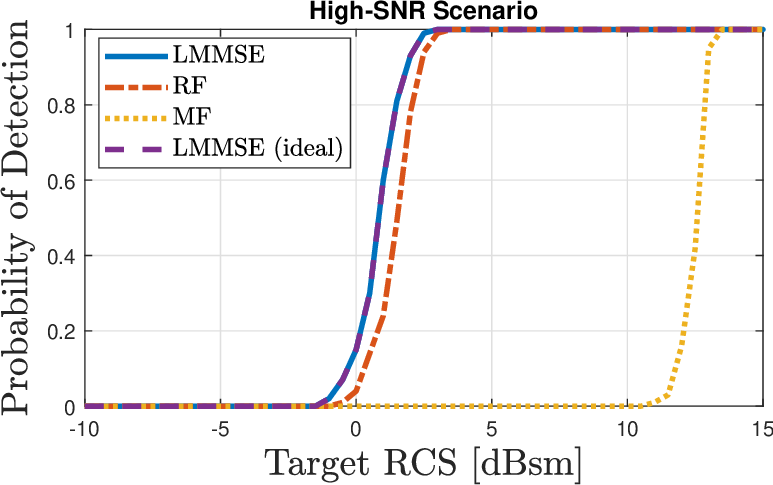}
		}
        \subfigure[]{
			 \label{fig_pd_vs_modOrder_threeTargets_Concurrent_highSNR}
			 \includegraphics[width=0.9\columnwidth]{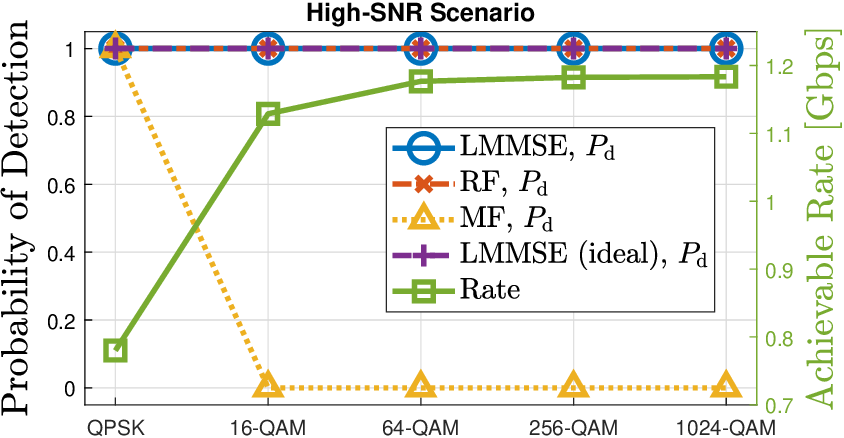}
		}
		
		\end{center}
		\vspace{-0.1in}
        \caption{ISAC performances in the high-SNR sensing scenario under concurrent transmission. \subref{fig_pd_vs_rcs_threeTargets_Concurrent_1024_QAM} Probability of detection with respect to RCS of Target-$1$ using $1024-$QAM data. \subref{fig_pd_vs_modOrder_threeTargets_Concurrent_highSNR} Probability of detection and achievable rate with respect to modulation order for Target-$1$ RCS of $5 \, \rm{dBsm}$.}  
        \label{fig_highSNR_TF_tradeoff}
        \vspace{0.1in}
\end{figure}


\vspace{-3mm}
\subsection{Spatial Trade-offs under Concurrent Transmission}\label{sec_num_spatial}
\vspace{-1mm}
To evaluate spatial domain ISAC trade-offs, we study the impact of the trade-off weight $\rho$ in \eqref{eq_ffm} on sensing and communication performances with varying modulation orders for $\boldX$, considering the same low-SNR and high-SNR scenarios as in Sec.~\ref{sec_tf_tradeoff}. The low-SNR scenario yields $\snrmean = (-60.3, -49.9) \, \rm{dB}$ for $\rho = (0, 1)$, respectively, while the high-SNR scenario leads to $\snrmean = (9.9, 20.3) \, \rm{dB}$ for $\rho = (0, 1)$. In the considered setting, the UE is located at an angle of $-30.2^\circ$, while the targets lie at $10^\circ$ and $18^\circ$. Hence, we expect to observe a trade-off between sensing and communications as $\rho$ varies. In Fig.~\ref{fig_QPSK_1024_QAM_spatial_tradeoff_curves}, we plot the ISAC trade-off curves obtained by the algorithms under consideration in the low-SNR and high-SNR scenarios as $\rho$ sweeps the interval $[0,1]$, using both QPSK and $1024-$QAM modulations.\footnote{Since all estimators are equivalent under QPSK signaling, as discussed in Sec.~\ref{sec_int_lmmse}, a single curve is plotted for the case of QPSK.}

\begin{figure}[t]
        \begin{center}
        \subfigure[]{
			 \label{fig_QPSK_1024_QAM_spatial_tradeoff_curves_lowSNR}
			 \includegraphics[width=0.9\columnwidth]{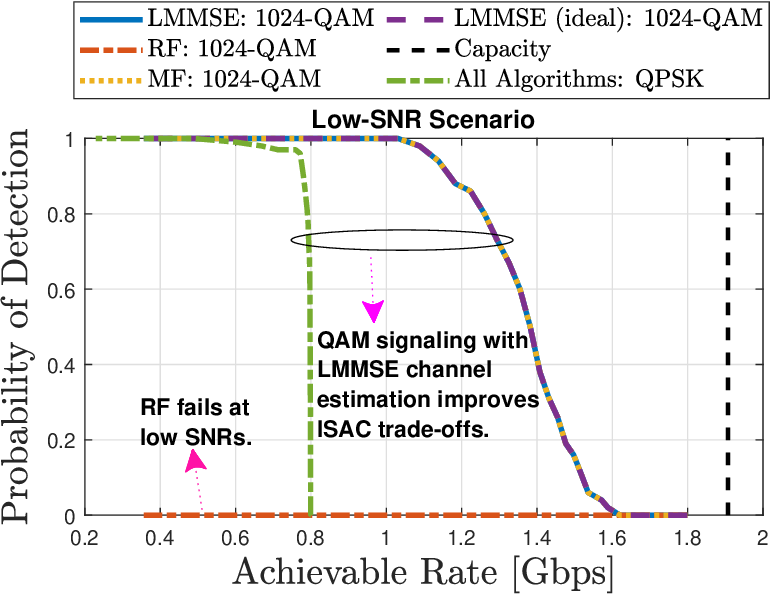}
		}
        \subfigure[]{
			 \label{fig_QPSK_1024_QAM_spatial_tradeoff_curves_highSNR}
			 \includegraphics[width=0.9\columnwidth]{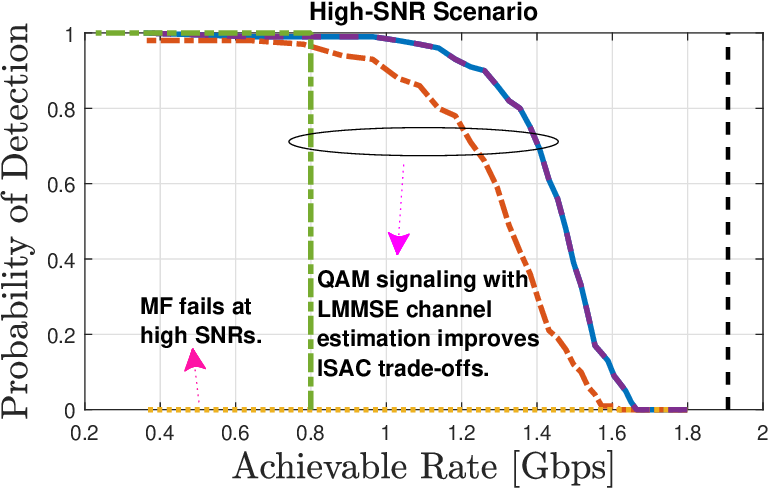}
		}
		
		\end{center}
		\vspace{-0.1in}
        \caption{ISAC trade-off curves under concurrent transmission for QPSK and $1024-$QAM modulations as the trade-off weight $\rho$ in \eqref{eq_ffm} varies over $[0,1]$. \subref{fig_QPSK_1024_QAM_spatial_tradeoff_curves_lowSNR} Low-SNR sensing scenario for target RCS of $-3 \, \rm{dBsm}$. \subref{fig_QPSK_1024_QAM_spatial_tradeoff_curves_highSNR} High-SNR sensing scenario for Target-$1$ RCS of $2 \, \rm{dBsm}$. }  
        \label{fig_QPSK_1024_QAM_spatial_tradeoff_curves}
\end{figure}

\subsubsection{Low-SNR Sensing Scenario}
Looking at the low-SNR results in Fig.~\ref{fig_QPSK_1024_QAM_spatial_tradeoff_curves_lowSNR}, LMMSE significantly improves ISAC trade-offs over RF  when utilizing $1024-$QAM signaling, aligning with insights from Table~\ref{tab_guideline} and Fig.~\ref{fig_pd_vs_rcs_singleTarget_Concurrent_1024_QAM}. In addition, as expected, the performance of LMMSE mirrors that of MF in low-SNR environments. When comparing QAM and QPSK signaling, we observe that LMMSE can achieve substantially better ISAC trade-offs with QAM signaling through its robustness in detection performance against increasing modulation order (as illustrated in Fig.~\ref{fig_pd_vs_modOrder_singleTarget_Concurrent}). This robustness enables the use of high-order QAM to boost rates without sacrificing sensing performance.\footnote{The steep trade-off curve for QPSK, in contrast to QAM, is attributed to the fact that the rate for QAM reaches its maximum at a higher SNR level than QPSK, leading to a more gradual decrease in $\pd$ as the rate increases. Note that as $\rho$ in \eqref{eq_ffm} approaches $0$, the sensing SNR decreases and the communication SNR increases.}

\subsubsection{High-SNR Sensing Scenario}
As observed in Sec.~\ref{sec_tf_tradeoff}, the trends for MF and RF become opposite when transitioning from low-SNR to high-SNR scenarios. Specifically, in high-SNR environments, MF fails to provide satisfactory ISAC trade-off performance whereas RF surpasses MF, albeit with a considerable performance gap compared to LMMSE. Hence, employing QAM signaling at the transmit side, in conjunction with the LMMSE algorithm at the sensing receiver, yields superior ISAC trade-offs over both RF and MF benchmarks, regardless of whether QAM or QPSK signaling is used.

\subsubsection{Summary of Low-SNR and High-SNR Scenarios}
In alignment with Table~\ref{tab_guideline} and the previous findings in Fig.~\ref{fig_lowSNR_TF_tradeoff} and Fig.~\ref{fig_highSNR_TF_tradeoff}, the spatial domain trade-off results in Fig.~\ref{fig_QPSK_1024_QAM_spatial_tradeoff_curves} demonstrate that LMMSE consistently outperforms the traditional OFDM radar sensing benchmarks MF and RF across various SNRs, modulation schemes and sensing-communication beamforming weights. Based on the findings in Fig.~\ref{fig_QPSK_1024_QAM_spatial_tradeoff_curves}, Table~\ref{tab_guideline2} provides rough guidelines on the choice of transmit signaling strategies and sensing channel estimators under different sensing SNRs and ISAC requirements.

\begin{table}
\caption{Guidelines on Which Transmit Signaling Strategy (QAM or QPSK) and Channel Estimator to Use Under Different SNR Regimes and ISAC Requirements}
\centering
    \begin{tabular}{|c|c|c|c|}
        \hline
        \textbf{ISAC}  & High $\pd$ & Medium $\pd$ & Low $\pd$ \\ 
        \textbf{Requirement} & Low Rate & Medium Rate & High Rate \\
        \hline
        \textit{Low}  & QPSK/ & QAM/ & QAM/ \\ 
        \textit{SNR} & All Estimators & LMMSE+MF & LMMSE+MF
        \\ \hline
        \textit{High} & QPSK/ & QAM/ & QAM/ \\ 
        \textit{SNR} & All Estimators & LMMSE & LMMSE \\
        \hline
    \end{tabular}
    \label{tab_guideline2}
   \vspace{0.1in}
\end{table}

\vspace{-4mm}
\subsection{Concurrent vs. Time-Sharing Transmission}

We now carry out a comparative analysis of Concurrent and Time-Sharing strategies, as depicted in Fig.~\ref{fig_isac_system} and described in \eqref{eq_ffm} and \eqref{eq_tst}, using the same scenarios in Sec.~\ref{sec_num_spatial}. Fig.~\ref{fig_conc_vs_timesharing} shows the trade-off curves achieved by Time-Sharing\footnote{Since only QPSK pilots are used for sensing in the Time-Sharing strategy (see Sec.~\ref{sec_tst} and Fig.~\ref{fig_isac_system}), all sensing algorithms are equivalent.} as the time-sharing ratio $\abss{\msens}/M$ in \eqref{eq_tst} sweeps the interval $[0,1]$, along with those belonging to the Concurrent strategy.

The comparison between Concurrent and Time-Sharing strategies reveals the effect of two counteracting factors, each prevailing under distinct operational conditions:
\begin{itemize}
    \item \textit{C-SNR-Boost: Full Data Utilization for Sensing:} Concurrent (C) strategy leverages the entire OFDM frame for sensing, offering a potential advantage in data utilization via SNR boosting over Time-Sharing, which limits sensing to dedicated pilots, reducing detection capabilities.
    \item \textit{NCM-Loss: Impact of Non-Constant-Modulus (NCM) QAM Data on Sensing:} While Concurrent transmission benefits from using QAM data for sensing, this does not guarantee enhanced detection over Time-Sharing. The presence of high side-lobe levels from QAM can be detrimental if the sensing algorithm is ill-equipped to handle them, placing Concurrent at a disadvantage against Time-Sharing, which relies on QPSK pilots for sensing.\footnote{It is worth emphasizing that the influence of these factors on sensing performance depends on the specific circumstances, suggesting that observed performance patterns may vary with changes in sensing and communication channel characteristics, time/power allocations between the two functionalities and the choice of sensing algorithm, as subsequent discussions will illustrate.}
\end{itemize}

\begin{figure}[t]
        \begin{center}
        \subfigure[]{
			 \label{fig_concurrent_vs_time_sharing_RCS_5_dBsm_lowSNR}
			 \includegraphics[width=0.9\columnwidth]{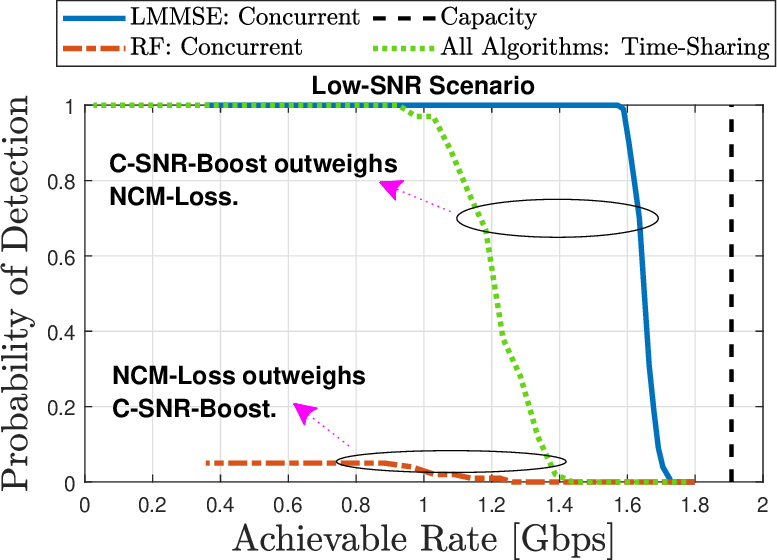}
		}
        \subfigure[]{
			 \label{fig_concurrent_vs_time_sharing_RCS_20_dBsm_highSNR}
			 \includegraphics[width=0.9\columnwidth]{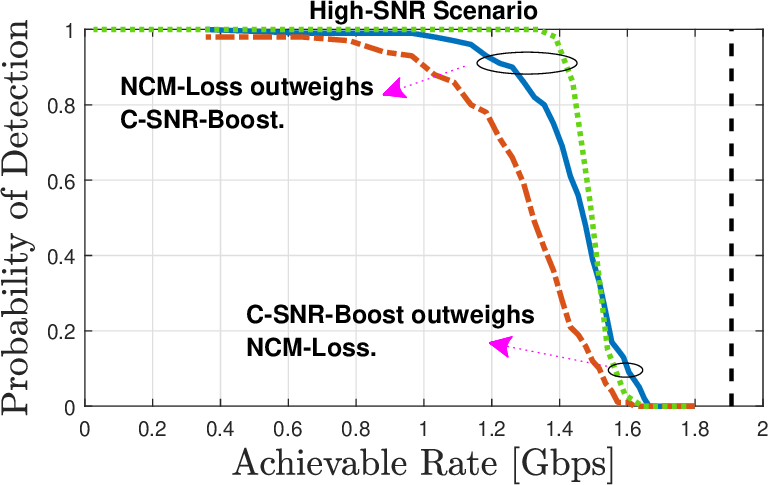}
		}
		
		\end{center}
		\vspace{-0.1in}
            \caption{ISAC trade-off curves under concurrent and time-sharing transmission strategies, where dedicated QPSK pilots are used for sensing in time-sharing and $1024-$QAM data are used for communications (see Fig.~\ref{fig_isac_system}). The time-sharing curves are obtained by sweeping $\abss{\msens}/M$ in \eqref{eq_tst} over $[0,1]$.  \subref{fig_concurrent_vs_time_sharing_RCS_5_dBsm_lowSNR} Low-SNR sensing scenario for target RCS of $0 \, \rm{dBsm}$. \subref{fig_concurrent_vs_time_sharing_RCS_20_dBsm_highSNR} High-SNR sensing scenario for target-$1$ RCS of $2 \, \rm{dBsm}$.}  
        \label{fig_conc_vs_timesharing}
\end{figure}

For the low-SNR scenario depicted in Fig.~\ref{fig_concurrent_vs_time_sharing_RCS_5_dBsm_lowSNR}, Concurrent transmission with QAM, in combination with LMMSE at the sensing receiver, yields significantly better ISAC trade-offs than Time-Sharing. This indicates that the full utilization of data for sensing (C-SNR-Boost) outweighs the challenges associated with using QAM data (NCM-Loss) in this scenario, confirming the effectiveness of the proposed LMMSE approach. The allocated sensing pilots in Time-Sharing are not sufficient to achieve $\pd$ as high as that achieved by Concurrent for an equivalent rate (corresponding to a certain $\abss{\msens}/M$ in Time-Sharing and a certain $\rho$ in Concurrent). Conversely, a scenario where NCM-Loss dominates C-SNR-Boost is seen when comparing the performance of the RF estimator in Concurrent transmission to Time-Sharing. Despite the utilization of the full OFDM frame with QAM data by the RF estimator in Concurrent transmission, it experiences a noticeable reduction in $\pd$ compared to Time-Sharing, which relies solely on QPSK pilots for sensing. This disadvantage for RF in Concurrent mode results from increased side-lobe levels at low SNRs, as previously shown in Fig.~\ref{fig_range_profile_RF_QPSK_QAM} and Fig.~\ref{fig_rangeProfile_CE_Comp}.

The high-SNR scenario in Fig.~\ref{fig_concurrent_vs_time_sharing_RCS_20_dBsm_highSNR} reveals trends in ISAC trade-offs that differ considerably from those observed in the low-SNR scenario, which highlights the scenario-specific nature of the relative weights of C-SNR-Boost and NCM-Loss. In this scenario, the high-SNR environment renders the use of dedicated QPSK pilots sufficiently effective for achieving a detection probability ($\pd$) of 1, making the full utilization of OFDM data for sensing (C-SNR-Boost) less critical up to a certain rate threshold. This effect is particularly noticeable until the rate approaches $1.4 \, \rm{Gbps}$, where $\pd$ of 1 is attainable with only a fraction of the transmit symbols dedicated to pilots. A key insight emerges when rates surpass $1.5 \, \rm{Gbps}$: Concurrent begins to outperform Time-Sharing as the reduced proportion of pilots no longer suffices for maintaining high $\pd$ in Time-Sharing. In contrast, exploitation of data in Concurrent provides additional SNR benefits, enhancing detection performance, thereby indicating the dominance of C-SNR-Boost over NCM-Loss in this rate regime. We note that such gains are possible with the use of the proposed LMMSE estimator, which can successfully counteract target masking effects, whereas RF (and, thus MF) with Concurrent transmission is outperformed by Time-Sharing. The threshold at which Concurrent gains an advantage over Time-Sharing varies with the scenario, influenced by the interplay between C-SNR-Boost and NCM-Loss. For instance, Concurrent outperforms Time-Sharing in all rate regimes in the low-SNR scenario.


\vspace{-3mm}
\subsection{Experimental Results}
In this part, to corroborate the findings in Fig.~\ref{fig_range_profile_RF_QPSK_QAM} and Fig.~\ref{fig_rangeProfile_CE_Comp}, we provide experimental results obtained via a full-duplex OFDM monostatic sensing setup at IMDEA Network laboratories. The data is collected using the Mimorph testbed from \cite{Lacruz_MOBISYS2021}, which was enhanced to support fully synchronized and concurrent transmit/receive operation. Millimeter wave front-end is formed by a Sivers IMA EVK operating at the 28~GHz frequency band. The front-end is composed by up/down converters and a $2  \times 8$ linear phased antenna array with analog beamforming capabilities. Experiments were collected in a indoor laboratory environment of $6 \, \rm{m} \times 8 \, \rm{m}$ with furniture. The testbed was located on one side of the room in the presence of a cylindrical metallic target, as shown in Fig.~\ref{fig_lab_setup}.

Fig.~\ref{fig_range_profile_experimental} shows the range profiles obtained by the considered algorithms using QPSK and $1024$-QAM modulations on experimentally collected data. We observe the self-interference peak resulting from full-duplex operation, along with the peak corresponding to the target at $2.76 \, \rm{m}$. Consistent with Fig.~\ref{fig_range_profile_RF_QPSK_QAM} and Fig.~\ref{fig_rangeProfile_CE_Comp}, the RF estimator produces substantially higher side-lobe levels with QAM modulation compared to QPSK modulation, particularly  beyond the $20 \, \rm{m}$ range. Moreover, the LMMSE estimator with QAM modulation (using $\snr_b = -20 \, \rm{dB}$ in \eqref{eq_lmmse_b}) can significantly reduce side-lobe levels and restore the range profile comparable to that achieved with QPSK modulation, in line with insights from Fig.~\ref{fig_rangeProfile_CE_Comp}. 


\begin{figure}
	\centering
	\includegraphics[width=0.8\linewidth]{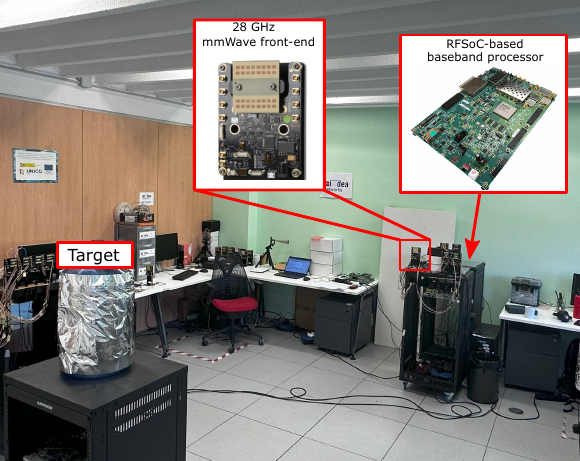}
	\vspace{-0.0in}
	\caption{Experimental setup with the testbed and a single target.} 
	\label{fig_lab_setup}
\end{figure}

\begin{figure}
	\centering
	\includegraphics[width=0.95\linewidth]{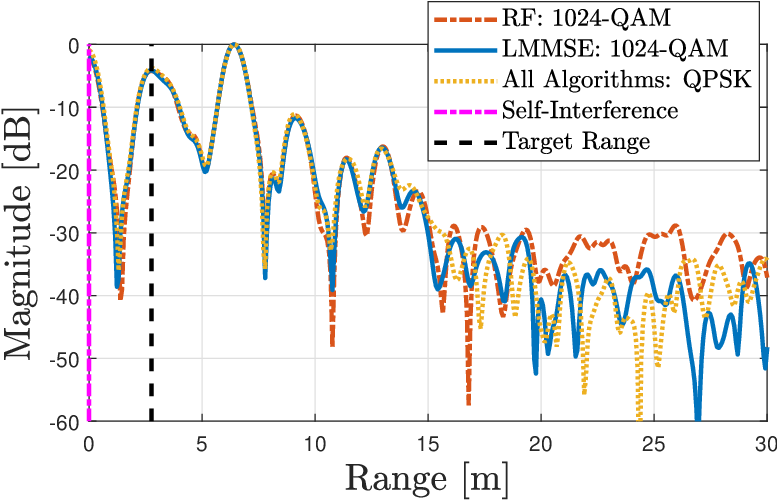}
	\caption{Range profiles obtained by different channel estimation algorithms with QPSK and $1024$-QAM modulations, using real-world experimental data collected via a monostatic full-duplex OFDM radar at $\fc = 28 \, \rm{GHz}$.} 
	\label{fig_range_profile_experimental}
\end{figure}

\vspace{-3mm}
\section{Concluding Remarks}
In this paper, we have investigated two fundamental trade-offs in monostatic OFDM ISAC systems: \textit{(i)} the \textit{time-frequency trade-off}, resulting from the choice of modulation scheme for random data, and \textit{(ii)} the \textit{spatial trade-off}, stemming from how ISAC transmit beamforming is applied to balance sensing and communication objectives. A novel LMMSE based sensing algorithm has been proposed to deal with the increased side-lobe levels induced by high-order QAM data, showcasing substantial improvements over existing OFDM radar sensing benchmarks. Moreover, two ISAC transmission strategies have been considered: \textit{(i)} \textit{Concurrent}, utilizing all data for sensing with power allocated between sensing and communication beams, and \textit{(ii)} \textit{Time-Sharing}, where sensing relies solely on dedicated pilots with time-multiplexing for sensing and communication beams. Extensive simulations demonstrate the superiority of LMMSE over the conventional RF and MF estimators under a wide range of operating conditions, while unveiling key insights into the ISAC trade-offs achieved by these strategies at different SNRs and modulation schemes. Moreover, two factors have been identified that influence the performance comparison between Concurrent and Time-Sharing, leading to scenario-dependent ISAC trade-off dynamic depending on the relative impact of these factors. As future research, we plan to focus on the same fundamental trade-offs in bistatic ISAC scenarios.
\vspace{-3.5mm}

\begin{appendices}

\section{Proof of Lemma~\ref{lemma_stat}}\label{app_lemma_stat}
\vspace{-1mm}
Using \eqref{eq_hhib_vec}, \eqref{eq_hhib_moments} and the definitions in Lemma~\ref{lemma_stat}, one can readily obtain $    \hhbar_{i,b} = \sum_{k=0}^{K-1} \Eee\{ \alpha_{b,k} \} \Eee\{\cc_b^\conj(\nu_k) \otimes \bb(\tau_k)  [\arx(\theta_k)]_i \} = \boldzero$.
As to the covariance, we have that
\begin{align} \label{eq_ccib_simple}
    \CC_{i,b} = \sum_{k_1=1}^K \sum_{k_2=1}^K \Eee\{ \alpha_{b,k_1} \alpha^\conj_{b,k_2}  \bphi_{k_1,b,i} \bphi^H_{k_2,b,i} \} \,, 
\end{align}
where $\bphi_{k,b,i} \triangleq \cc_b^\conj(\nu_k) \otimes \bb(\tau_k)  [\arx(\theta_k)]_i$. Due to the independence across the different parameters and the different targets, and using the definitions in Lemma~\ref{lemma_stat}, \eqref{eq_ccib_simple} simplifies to 
\begin{align} \label{eq_ccib_simple2}
    \CC_{i,b} &= \sum_{k=0}^{K-1}  \Eee\{ \abss{\alpha_{b,k}}^2 \}   \Eee\{\bphi_{k,b,i} \bphi^H_{k,b,i} \} \, ,
    \\ \nonumber
     &= \sum_{k=0}^{K-1}  \sigmaakb^2 \Eee\big\{ \cc_b^\conj(\nu_k) \cc_b^T(\nu_k) \big\} \otimes  \Eee\big\{ \bb(\tau_k)  \bb^H(\tau_k) \big\} \,,
\end{align}
where 
the mixed-product property of the Kronecker product is used. Using \eqref{eq_steer_delay}, it follows that $ \Eee\big\{ [\bb(\tau_k)  \bb^H(\tau_k)]_{n,n} \big\} = 1 $ and $\Eee\big\{ [\bb(\tau_k)  \bb^H(\tau_k)]_{n_1,n_2} \big\} = \int_{0}^{\frac{1}{\deltaf}} e^{j 2\pi (n_2-n_1) \deltaf \tau_k} \d \tau_k = 0$, 
for $n_1 \neq n_2$, which yields $\Eee\big\{ \bb(\tau_k)  \bb^H(\tau_k) \big\} = \Imatrix_N$. Similarly, using \eqref{eq_steer_doppler}, we obtain $\Eee\big\{ \cc_b^\conj(\nu_k) \cc_b^T(\nu_k)  \big\} = \Imatrix_{M_b}$. Plugging these into \eqref{eq_ccib_simple2} gives $\CC_{i,b} = \sigmaaab^2\Imatrix$.


\end{appendices}

\bibliographystyle{IEEEtran}
\bibliography{IEEEabrv,Sub/main}

\end{document}